\newcommand{\shortversion}[1]{#1}
\newcommand{\longversion}[1]{}

\documentclass[10pt]{article}
\usepackage{fullpage}
\usepackage{mdwlist}			
\usepackage{amssymb,amsmath}
\usepackage{amsthm}
\usepackage[usenames,dvipsnames]{color}
\usepackage[colorlinks=true,linkcolor=blue,citecolor=Mahogany]{hyperref}
\usepackage[poorman]{cleveref}
\usepackage[comma,numbers,square,sort]{natbib} 
\usepackage{multirow}
\usepackage{mathtools}
\usepackage{nicefrac}

\usepackage{thmtools, thm-restate}

\usepackage{xspace}
\newcommand{\ignore}[1]{}

\DeclareMathOperator*{\argmax}{arg\,max}

\newcommand{\el}{\ell}

\usepackage{makecell}

\usepackage{enumitem}% http://ctan.org/pkg/enumitem

\usepackage[usenames,svgnames,dvipsnames]{xcolor}

\usepackage{wrapfig}
\usepackage{lipsum}
\ignore{
\shortversion{

\usepackage{titlesec}
\titlespacing{\section}{0.5ex}{0.5ex}{0.5ex}
\titlespacing{\subsection}{0.5ex}{0.5ex}{0.5ex}
\titlespacing{\subsubsection}{0.5ex}{0.5ex}{0.5ex}
\titlespacing{\paragraph}{0pt}{0pt}{0pt}

\setlength{\topsep}{1ex}
\setlength{\partopsep}{0pt plus 0pt minus 0pt}
\setlength{\parskip}{0pt}
\usepackage{titling}
\posttitle{\par\end{center}}
\setlength{\droptitle}{-70pt}
\setlength{\textfloatsep}{0ex}
\setlength{\floatsep}{0ex}
\setlength{\intextsep}{0ex}

\AtBeginDocument{%
\abovedisplayskip=0pt plus 0pt minus 0pt
\abovedisplayshortskip=0pt plus 0pt
\belowdisplayskip=0pt plus 3pt minus 0pt
\belowdisplayshortskip=7pt plus 3pt minus 4pt
}
}}
\usepackage{authblk}

\usepackage{hyperref}
\usepackage{url}

\newcommand{\E}[1]{\mathbb{E}\left[#1\right]}

\newtheorem{proposition}{\bf Proposition}

\newtheorem{theorem}{\bf Theorem}
\newtheorem{lemma}{\bf Lemma}

\newtheorem{definition}{\bf Definition}
\newtheorem{claim}{\bf Claim}

\newcommand{\nfrac}{\nicefrac}

\renewcommand{\phi}{\varphi}

\newcommand{\eps}{\varepsilon}
\renewcommand{\epsilon}{\eps}

\usepackage[backgroundcolor=Moccasin,colorinlistoftodos]{todonotes}

\makeatletter
 \providecommand\@dotsep{5}
 \def\listtodoname{}
 \def\listoftodos{\@starttoc{tdo}\listtodoname}
 \makeatother

\newcounter{pdcomment}

\newcounter{abcomment}

\newcounter{dwcomment}

\newcommand{\UC}{\mathcal{U}}

\newcommand{\DC}{\mathcal{D}}

\newcommand{\SC}{\mathcal{S}}
\newcommand{\TC}{\mathcal{T}}
\newcommand{\BC}{\mathcal{B}}
\newcommand{\HC}{\mathcal{H}}
\newcommand{\LC}{\mathcal{L}}
\newcommand{\EC}{\mathcal{E}}
\newcommand{\RC}{\mathcal{R}}
\newcommand{\IC}{\mathcal{I}}
\newcommand{\NB}{\mathbb{N}}
\newcommand{\EB}{\mathbb{E}}

\newcommand{\vv}{\mathfrak{v}}

\usepackage[noend]{algpseudocode}
\usepackage{algorithm,algorithmicx}

\newcommand*\Let[2]{\State #1 $\gets$ #2}
\algrenewcommand\algorithmicrequire{\textbf{Input:}}
\algrenewcommand\algorithmicensure{\textbf{Output:}}
\algnewcommand{\Initialize}[1]{%
  \State \textbf{Initialize:}
  \Statex \hspace*{\algorithmicindent}\parbox[t]{0.9\linewidth}{\raggedright #1}
}

\usepackage{cleveref}
\renewcommand{\cref}{\Cref}

\crefname{theorem}{theorem}{\bf Theorem}
\crefname{observation}{observation}{\bf Observation}
\crefname{lemma}{lemma}{\bf Lemma}
\crefname{corollary}{corollary}{\bf Corollary}
\crefname{proposition}{proposition}{\bf Proposition}
\crefname{definition}{definition}{\bf Definition}
\crefname{claim}{claim}{\bf Claim}
\crefname{table}{table}{\bf Table}
\crefname{reductionrule}{reduction rule}{\bf Reduction rule}
\crefname{ALC@unique}{line}{Line}

\renewcommand{\leq}{\leqslant}
\renewcommand{\geq}{\geqslant}
\renewcommand{\ge}{\geqslant}
\renewcommand{\le}{\leqslant}

\newcommand{\lhh}{\textsc{$(\eps,\varphi)$-List heavy hitters}\xspace}

\title{An Optimal Algorithm for $\ell_1$-Heavy Hitters in Insertion Streams and Related Problems}

\author{Arnab Bhattacharyya$^\star$, Palash Dey$^\star$, and David P. Woodruff$^\dag$}
\affil{$^\star$Indian Institute of Science, Bangalore}
\affil{$^\dag$IBM Research, Almaden}
\affil{\texttt {$^\star$\{arnabb,palash\}@csa.iisc.ernet.in,$^\dag$dpwoodru@us.ibm.com}}

\date{}

\begin{document}

\maketitle

\begin{abstract}
We give the first optimal bounds for returning the $\ell_1$-heavy hitters in a data stream of insertions, together with their approximate frequencies, closing a long line of work on this problem. For a stream of $m$ items in $\{1, 2, \ldots, n\}$ and parameters $0 < \epsilon < \phi \leq 1$, let $f_i$ denote the frequency of item $i$, i.e., the number of times item $i$ occurs in the stream. With arbitrarily large constant probability, our algorithm returns all items $i$ for which $f_i \geq \phi m$, returns no items $j$ for which $f_j \leq (\phi -\epsilon)m$, and returns approximations $\tilde{f}_i$ with $|\tilde{f}_i - f_i| \leq \epsilon m$ for each item $i$ that it returns. Our algorithm uses $O(\epsilon^{-1} \log\phi^{-1} + \phi^{-1} \log n + \log \log m)$ bits of space, processes each stream update in $O(1)$ worst-case time, and can report its output in time linear in the output size. We also prove a lower bound, which implies that our algorithm is optimal up to a constant factor in its space complexity. A modification of our algorithm can be used to estimate the maximum frequency up to an additive $\epsilon m$ error in the above amount of space, resolving Question 3 in the IITK 2006 Workshop on Algorithms for Data Streams for the case of $\ell_1$-heavy hitters. We also introduce several variants of the heavy hitters and maximum frequency problems, inspired by rank aggregation and voting schemes, and show how our techniques can be applied in such settings. Unlike the traditional heavy hitters problem, some of these variants look at comparisons between items rather than numerical values to determine the frequency of an item. 
\end{abstract}

\section{Introduction}
The data stream model has emerged as a standard model for
processing massive data sets. Because of the sheer size
of the data, traditional algorithms are no longer
feasible, e.g., it may be hard or impossible to store the
entire input, and algorithms need to run in linear or even
sublinear time. Such algorithms typically need to be both
randomized and approximate. Moreover, the data may not
physically reside on any device, e.g., if it is internet
traffic, and so if the data is not stored by the algorithm,
it may be impossible to recover it. Hence, many algorithms
must work 
given only a single pass over the data. Applications
of data streams include data warehousing \cite{HSST05,br99,fsgmu98,HPDW01},
network measurements \cite{ABW03,GKMS08,demaine2002frequency,ev03}, 
sensor networks \cite{BGS01,SBAS04},
and compressed sensing \cite{GSTV07,CRT05}. We refer the reader
to recent surveys on the data stream model 
\cite{muthukrishnan2005data,nelson2012sketching, Cormode2012}.

One of the oldest and most fundamental problems in the area of data
streams is the problem of finding the $\ell_1$-heavy hitters (or simply, ``heavy hitters''), also known as the
top-$k$, most popular items, frequent items, elephants, or iceberg queries.
Such algorithms can be used as subroutines in network
flow identification at IP routers \cite{ev03}, association
rules and frequent itemsets \cite{as94,son95,toi96,hid99,hpy00}, 
iceberg queries 
and iceberg datacubes \cite{fsgmu98,br99,HPDW01}. The survey \cite{cormode2008finding} presents 
an overview of the state-of-the-art for this problem, from both theoretical and practical standpoints.

We now formally define the heavy hitters problem that we
focus on in this paper: 
\begin{definition}{\bf ($(\epsilon, \phi)$-Heavy Hitters Problem)}\label{def:hh}
In the $(\epsilon, \phi)$-Heavy Hitters Problem, we are given parameters
$0 < \epsilon < \phi \leq 1$ and 
a stream $a_1, \ldots, a_m$ of items $a_j \in \{1, 2, \ldots, n\}.$
Let $f_i$ denote the number of occurrences of item $i$, i.e., its 
frequency.  
The algorithm should make one pass over the stream and at the
end of the stream output a set $S \subseteq \{1, 2, \ldots, n\}$
for which if $f_i \geq \phi m$, then $i \in S$, while if 
$f_i \leq (\phi - \epsilon)m$, then $i \notin S$. Further, for
each item $i \in S$, the algorithm should output an estimate
$\tilde{f}_i$ of the frequency $f_i$ which satisfies
$|f_i - \tilde{f}_i| \leq \epsilon m$. 
\end{definition}
Note that other natural definitions of heavy hitters are possible and sometimes used. For example, {\em $\ell_2$-heavy hitters} are those items $i$ for which $f_i^2 \geq \phi^2 \sum_{j=1}^n f_j^2$, and more generally, {\em $\ell_p$-heavy hitters} are those items $i$ for which $f_i^p \geq \phi^p \sum_{j=1}^n f_j^p$. It is in this sense that Definition \ref{def:hh} corresponds to { $\ell_1$-heavy hitters}. While $\ell_p$-heavy hitters for $p>1$ relax $\ell_1$-heavy hitters and algorithms for them have many interesting applications, we focus on the most direct and common formulation of the heavy hitters notion.

We are
interested in algorithms which use as little space in bits
as possible to solve the {\bf $(\epsilon, \phi)$-Heavy Hitters Problem}.
Further, we are also interested in minimizing the {\it update time}
and {\it reporting time} of such algorithms. Here, the update time
is defined to be the time the algorithm needs to update its data
structure when processing a stream insertion. The reporting time is the
time the algorithm needs to report the answer after having 
processed the stream. We allow the algorithm to be randomized and
to succeed with probability at least $1-\delta$ for $0<\delta<1$.
We do not make any assumption on the ordering of the stream
$a_1, \ldots, a_m$. This
is desirable as often in applications one cannot assume a best-case or even
a random order. We are also interested in the case when the length $m$ of the stream is
not known in advance, and give algorithms in this more general 
setting. 

The first algorithm for the {\bf $(\epsilon, \phi)$-Heavy Hitters Problem}
was given by Misra and Gries \cite{misra82}, who achieved $O(\epsilon^{-1} (\log n + \log m))$
bits of space for any $\phi > \epsilon$. This algorithm was rediscovered
by Demaine et al. \cite{demaine2002frequency}, and again by Karp et al. \cite{karp2003simple}. 
Other than these 
algorithms, which are deterministic, there are also a number of randomized
algorithms, such as the CountSketch \cite{charikar2004finding}, Count-Min sketch \cite{cormode2005improved}, sticky sampling \cite{mm02},
lossy counting \cite{mm02}, space-saving \cite{MetwallyAA05}, sample and hold \cite{ev03}, multi-stage bloom filters
\cite{cfm09}, and sketch-guided sampling \cite{kx06}. 
Berinde
et al. \cite{bics10} show that using $O(k \epsilon^{-1} \log(mn))$ bits of space, 
one can achieve the stronger guarantee
of reporting, for each item $i \in S$, $\tilde{f}_i$ with
$|\tilde{f}_i - f_i| \leq \nfrac{\epsilon}{k} F^{res(k)}_1$, 
where $F^{res(k)}_1 < m$ denotes
the sum of frequencies of items in $\{1, 2, \ldots, n\}$ excluding the frequencies
of the $k$ most frequent items. 

We emphasize that prior to our work the
best known algorithms for the {\bf $(\epsilon, \phi)$-Heavy Hitters Problem}
used $O(\epsilon^{-1} (\log n + \log m))$ bits of space. 
Two previous lower bounds were known. The first is a lower bound of $\log({n \choose 1/\phi})
= \Omega(\phi^{-1} \log(\phi n))$ bits, which  comes from the fact that the output set $S$ can contain $\phi^{-1}$
items and it takes this many bits to encode them. 
The second lower bound is
$\Omega(\epsilon^{-1})$ 
which follows from a folklore 
reduction from the randomized communication complexity of the {\sf Index} problem. 
In this problem, there are two players, Alice and Bob. Alice has a bit string $x$ of length $(2\epsilon)^{-1}$,
while Bob has an index $i$. Alice creates a stream of length $(2\eps)^{-1}$ consisting of one copy of each $j$ for which $x_j = 1$ and copies of a dummy item to fill the rest of the stream. She runs the heavy hitters streaming algorithm on her stream and sends the state of the
algorithm to Bob. Bob appends $(2\epsilon)^{-1}$ 
copies of the item $i$ to the stream 
and continues the execution
of the algorithm. For $\phi = \nfrac{1}{2}$, it holds that $i \in S$. Moreover, $f_i$ differs
by an additive $\epsilon m$ factor depending on whether $x_i = 1$ or $x_i = 0$. 
Therefore by the randomized communication complexity of the {\sf Index} problem 
\cite{kremer1999randomized}, the 
$(\epsilon, \nfrac{1}{2})$-heavy hitters problem requires $\Omega(\epsilon^{-1})$ bits of space.
Although this proof was for $\phi = \nfrac{1}{2}$, no better lower bound is known for any
$\phi > \epsilon$. 

Thus, while the upper bound for the {\bf $(\epsilon, \phi)$-Heavy Hitters Problem} 
is $O(\epsilon^{-1} (\log n + \log m))$ bits, the best known 
lower bound is only $\Omega(\phi^{-1} \log n + \epsilon^{-1})$ bits. For constant
$\phi$, and $\log n \approx \epsilon^{-1}$, this represents a nearly quadratic gap
in upper and lower bounds. Given the limited resources of devices which typically
run heavy hitters algorithms, such as internet routers, this quadratic gap can 
be critical in applications. 
%Moreover, given the number of problems in the data stream
%literature which use a solution to the {\bf $(\epsilon, \phi)$-Heavy Hitters Problem} 
%as a subroutine, improving
%this bound can improve a wide number of bounds in the literature. 

A problem related to the {\bf $(\epsilon, \phi)$-Heavy Hitters Problem} is estimating
the {\it maximum frequency} in a data stream, also known as the $\ell_{\infty}$-norm. In
the IITK 2006 Workshop on Algorithms for Data Streams, Open Question 3 asks for an algorithm
to estimate the maximum frequency of any item up to an additive $\epsilon m$ error using
as little space as possible. The best known space bound is still $O(\epsilon^{-1} \log n)$
bits, as stated in the original formulation of the question (note that the ``$m$'' 
in the question there corresponds to the ``$n$'' here). Note that, if one can find an 
item whose frequency is the largest, up to an additive $\epsilon m$ error, then one can solve
this problem. The latter problem is independently interesting and corresponds to finding approximate
plurality election winners in voting streams \cite{deysampling}. 
We refer to this problem as the \textbf{$\epsilon$-Maximum} problem. 

Finally, we note that there are many other variants of the 
{\bf $(\epsilon, \phi)$-Heavy Hitters Problem} that one can consider. One simple variant
of the above is to output an item of frequency within $\epsilon m$ of the {\it minimum frequency} 
of any item in the universe. We refer to this as the 
\textbf{$\epsilon$-Minimum} problem. 
This only makes sense for small universes, as otherwise outputting
a random item typically works. This is useful when one wants to count
the ``number of dislikes'', or in anomaly detection; see more motivation
below. In other settings, 
one may not have numerical scores associated with the items, but rather, 
each stream update consists of a ``ranking'' or ``total ordering'' of all stream items. 
This may be the case in ranking aggregation on the web (see, e.g., \cite{mbg04,MYCC07}) 
or in voting streams (see, e.g., \cite{conitzer2005communication, caragiannis2011voting, deysampling, xia2012computing}). 
One may consider a variety of aggregation measures, such as the {Borda score} of 
an item $i$, which asks for the sum, over rankings, of the number of items $j \neq i$ for
which $i$ is ranked ahead of $j$ in the ranking. Alternatively, 
one may consider the { Maximin score}
of an item $i$, which asks for the minimum, over items $j \neq i$, of the number
of rankings for which $i$ is ranked ahead of $j$. For these aggregation measures, one may
be interested in finding an item whose score is an approximate maximum. 
This is the analogue of the 
\textbf{$\epsilon$-Maximum} problem above. Or, one may be interested in listing
all items whose score is above a threshold, which is the analogue of the 
{\bf $(\epsilon, \phi)$-Heavy Hitters Problem}. 
%We refer to these problems corresponding to Borda and maximin as the {\bf $\epsilon$-Borda} and {\bf $\epsilon$-Maximin} problems respectively, where
%the nature of the output and approximation will be specified from the context. 

We give more motivation of these variants of heavy hitters in this section below, and more precise definitions in Section \ref{sec:prelim}. 
\begin{table*}[t]
 \begin{center}
  \resizebox{\textwidth}{!}{%\renewcommand{\arraystretch}{2.5}
   \begin{tabular}{|c|c|c|}\hline
   
      \multirow{2}{*}{\textbf{Problem}}	& \multicolumn{2}{c|}{\textbf{Space complexity}} \\\cline{2-3}
       & Upper bound & Lower bound \\\hline\hline
      
      &&\\[-10pt]
      
      $(\epsilon, \phi)$-Heavy Hitters & \makecell{$O\left( \epsilon^{-1} \log \phi^{-1} + \phi^{-1} \log n + \log \log m \right )$ \\~[\Cref{thm:heavy_hitters,thm:UbUnknownMax}]}
      & 
      \makecell{$\Omega \left (\epsilon^{-1} \log \phi^{-1} + \phi^{-1} \log n + \log \log m \right )$\\~[\Cref{thm:eps_eps,thm:loglogn}]}
      
      \\\hline
      
      &&\\[-10pt]
      
      $\epsilon$-Maximum and $\ell_{\infty}$-approximation& \makecell{$O\left(\epsilon^{-1} \log \epsilon^{-1} + \log n + \log \log m \right )$\\~[\Cref{thm:heavy_hitters,thm:UbUnknownMax}]} &
      \makecell{$\Omega \left (\epsilon^{-1} \log \epsilon^{-1} + \log n + \log \log m \right )$\\~[\Cref{thm:eps_eps,thm:loglogn}]}
      \\\hline
      
      &&\\[-10pt]
      
      $\epsilon$-Minimum &  \makecell{$O\left(\epsilon^{-1} \log \log \epsilon^{-1} + \log \log m \right)$ \\~[\Cref{thm:rare,thm:UbUnknownMin}]}
      & \makecell{$\Omega \left (\epsilon^{-1} + \log \log m \right )$\\~[\Cref{thm:veto_lb,thm:loglogn}]}\\ \hline
      &&\\[-10pt]
      
  $\epsilon$-Borda &  \makecell{$O\left(n(\log \epsilon^{-1} + \log n) + \log \log m \right)$ \\~[\Cref{thm:borda,thm:UbUnknownMin}]}
      & \makecell{$\Omega \left (n (\log\epsilon^{-1} + \log n) + \log \log m \right )$\\~[\Cref{thm:lwb_borda,thm:loglogn}]}\\ \hline
      
  $\epsilon$-Maximin &  \makecell{$O\left(n \epsilon^{-2}\log^2 n + \log \log m \right)$ \\~[\Cref{thm:maximin,thm:UbUnknownMin}]}
      & \makecell{$\Omega \left (n (\epsilon^{-2} + \log n) + \log \log m \right )$\\~[\Cref{thm:mmlb}]}\\ \hline
   \end{tabular}
  }
  \caption{\small 
  %Our space complexity bounds in bits for the problems considered in this paper. 
  The bounds hold for constant success probability algorithms and for $n$ sufficiently large in terms of $\eps$. For the {\bf $(\epsilon, \phi)$-Heavy Hitters}
problem and the {\bf $\epsilon$-Maximum} problem, we also
achieve $O(1)$ update time and reporting time which is linear in the size
of the output. The upper bound for
{\bf $\epsilon$-Borda} (resp. {\bf $\epsilon$-Maximin})
is for returning every item's Borda score (resp. Maximin score)
up to an additive $\epsilon mn$ (resp. additive $\epsilon m$), 
while the lower bound for {\bf $\epsilon$-Borda} 
(resp. {\bf $\epsilon$-Maximin}) is for
returning only the approximate Borda score (resp. Maximin score) 
of an approximate maximum.}\label{tbl:summary}
 \end{center}
\end{table*}

\subsection{Our Contributions}
Our results are summarized in Table \ref{tbl:summary}. We note that independently of this work and nearly parallelly, there have been improvements to the space complexity of the $\ell_2$-heavy hitters problem in insertion streams \cite{BCIW16} and to the time complexity of the $\ell_1$-heavy hitters problem in turnstile\footnote{In a turnstile stream, updates modify an underlying $n$-dimensional vector $x$ initialized at the zero vector; each update is of the form $x \leftarrow x+ e_i$ or $x \leftarrow x-e_i$ where $e_i$ is the $i$'th standard unit vector. In an insertion stream, only updates of the form $x \leftarrow x+ e_i$  are allowed.} streams \cite{LNNT16}. These works use very different techniques.

Our first contribution is an optimal algorithm and lower bound 
for the {\bf $(\epsilon, \phi)$-Heavy Hitters Problem}. 
Namely, we show that there
is a randomized algorithm with constant probability of success which solves this problem using 
$$O(\epsilon^{-1} \log \phi^{-1} + \phi^{-1} \log n + \log \log m)$$ 
bits of space, and we prove a lower bound matching up to constant factors. 
In the unit-cost RAM model with $O(\log n)$ bit words, 
our algorithm has $O(1)$ update time
and reporting time linear in the output size, under the
standard assumptions that the length of the
stream and universe size are at least $\text{poly}(\eps^{-1} \log(1/\phi))$. Furthermore, we can achieve nearly the optimal space complexity even when the
length $m$ of the stream is {\it not known in advance}.
Although the results of \cite{bics10} achieve stronger error bounds in terms of the tail, which are useful for skewed streams, here we focus on the original formulation of the problem.
%Our upper and lower bounds make the mild assumption that the universe size 
%$n \geq (1/\eps)^{1+\mu}$ for an arbitrary constant $\mu > 0$ 
\ignore{We also analyze the case of $n \leq 1/\eps$ for which we show 
that the space complexity is $\Theta(n \log \phi^{-1} + \phi^{-1}\log n + \log \log m)$. }

Next, we turn to the problem of estimating the maximum frequency in a data
stream up to an additive $\epsilon m$. We give an algorithm using
$$O(\epsilon^{-1} \log \epsilon^{-1} + \log n + \log \log m)$$
bits of space, improving the previous best algorithms which required
space at least $\Omega(\epsilon^{-1} \log n)$ bits, and show that our bound is tight. As an example setting of parameters, if $\epsilon^{-1} = \Theta(\log n)$ and $\log \log m = O(\log n)$, our space
complexity is $O(\log n \log \log n)$ bits, improving the previous 
$\Omega(\log^2 n)$ bits of space algorithm. We also prove a lower bound
showing our algorithm is optimal up to constant factors. 
This resolves Open Question 3 from the IITK 2006 Workshop on Algorithms for Data Streams in the case of insertion streams, for the case of ``$\ell_1$-heavy hitters''. Our algorithm also returns the identity of the item with the 
approximate maximum frequency, solving the \textbf{$\epsilon$-Maximum} 
problem. \ignore{Again, our lower bound assumes that $n \geq (1/\eps)^{1+\mu}$ for arbitrary constant $\mu > 0$ but we can show a tight result for the case $n \leq 1/\eps$ also.}

We then focus on a number of variants of these problems. 
We first give nearly tight bounds for finding an item whose frequency
is within $\epsilon m$ of the minimum possible frequency. While this can be solved
using our new algorithm for the {\bf $(\epsilon, \epsilon)$-Heavy Hitters Problem},
this would incur $\Omega(\epsilon^{-1} \log \epsilon^{-1} + \log\log m)$ bits of space, whereas we give
an algorithm using only $O(\epsilon^{-1} \log \log (\epsilon^{-1}) + \log\log m)$ bits of space.  We
also show a nearly matching $\Omega(\epsilon^{-1} + \log\log m)$ bits of space lower bound. We note
that for this problem, a dependence on $n$ is not necessary since if the number of
possible items is sufficiently large, then outputting the identity
of a random item among the first say, $10\epsilon^{-1}$ items, is a correct solution
with large constant probability. 

Finally, we study variants of heavy hitter problems that are ranking-based. In this setting, each
stream update consists of a total ordering of the $n$ universe items. For the $\epsilon$-{\bf Borda}
problem, we give an algorithm using $O(n (\log \epsilon^{-1} + \log \log n) + \log \log m)$
bits of space to report the Borda score of every item up to an additive $\epsilon m n$. 
We also show this is nearly optimal by proving an $\Omega(n \log \epsilon^{-1} + \log\log m)$ bit lower bound for the problem, even in the case when one is only interested in outputting an item maximum Borda score up to an additive $\epsilon m n$. For the $\epsilon$-{\bf Maximin} problem, we give an algorithm using $O(n \epsilon^{-2} \log^2n + \log \log m)$ bits of space to report the maximin score of every item up to an additive $\epsilon m$, and prove an $\Omega(n \epsilon^{-2} + \log\log m)$ bits of space lower bound even in the case when one is only interested in outputting the maximum maximin
score up to an additive $\epsilon m$. 
This shows that finding heavy hitters
with respect to the maximin score is significantly more expensive than with respect to
the Borda score.

\subsection{Motivations for Variants of Heavy Hitters}
While the {\bf $(\epsilon, \phi)$-Heavy Hitters} and
{\bf $\epsilon$-Maximum} problem are very well-studied in the data stream
literature, the other variants introduced are not. We provide additional
motivation for them here. 

For the {\bf $\epsilon$-Minimum} problem, in our formulation, an item with
frequency zero, i.e., one that does not occur in the stream, is a valid solution
to the problem. In certain scenarios, this might not make sense, e.g., if a stream containing only a small fraction of IP addresses. However, in other scenarios we argue this is a natural problem.
For instance, consider an online portal where users register complaints
about products. Here, minimum frequency items correspond to the ``best'' items.
That is, such frequencies arise in the context of voting or more generally
making a choice: in cases for which one does not have a strong preference
for an item, but definitely does not like certain items, this problem
applies, since the frequencies correspond to ``number of dislikes''. 

The {\bf $\epsilon$-Minimum} problem may also be useful for anomaly detection.
Suppose one has a known set of sensors broadcasting information and one
observes the ``From:'' field in the broadcasted packets. Sensors which send
a small number of packets may be down or defective, and an algorithm
for the {\bf $\epsilon$-Minimum} problem could find such sensors. 
%
% is important for finding rare items, which is related to the rarity problem and examples studied in \cite{dm02}. For example, in network traffic, a rare IP address could indicate a user who visited a website for only a short period of time, possibly because the user was unhappy with his/her customer experience. Another example is denial of service attacks - if a flow of IP traffic only has a single packet (the so-called SYN packet), and there are many such flows, it could indicate a denial of service attack. Finding an $\epsilon$-Minimum could give the identity of one such flow. 

Finding items with maximum and minimum frequencies in a stream correspond to finding winners under plurality and veto voting rules respectively in the context of voting\footnote{In fact, the first work \cite{Moore81} to formally pose the heavy hitters problem couched it in the context of voting.} \cite{brandt2015handbook}. The streaming aspect of voting could be crucial in applications like online polling~\citep{kellner2011polling}, recommender systems~\citep{resnick1997recommender,herlocker2004evaluating,adomavicius2005toward} where the voters are providing their votes in a streaming fashion and at every point in time, we would like to know the popular items.
   While in some elections, such as for political positions, the scale of
the election may not be large enough to require a streaming algorithm, one
key aspect of these latter voting-based problems is that they are rank-based
which is useful when numerical scores are not available. Orderings naturally
arise in several applications - for instance, if a website has multiple
parts, the order in which a user visits the parts given by its clickstream
defines a voting, and for data mining and recommendation purposes the
website owner may be interested in aggregating the orderings across users. 
% The rules we consider are natural ways of aggregating such orderings.
%Indeed, for example, consider an online portal where users rate items. Here, one may like to find the items with maximum total ratings which corresponds to items with maximum frequency (also known as plurality winners in underlying election) in the stream where a rating of $t$ for an item $i$ is replaced with $t$ many copies of $i$ in the stream. Similarly, consider an online portal where users register complaints about products. Here, one may like to find the best products in the sense of receiving minimum number of complaints which corresponds to finding items with minimum frequencies (also known as veto winners in the social choice literature).
 Motivated by this connection, we define similar problems for two other important voting rules, namely Borda and maximin. 
%  The Borda measure is used as a voting rule in the National Assembly of Slovenia, Icelandic parliamentary elections, and is similar to that used in the Eurovision song contest. The Borda score sometimes results in the output of an election being a broadly acceptable candidate rather than the one preferred by a majority, and so is sometimes described as a ``consensus-based" voting system rather than a majoritarian one. 
 The Borda scoring method finds its applications in a wide range of areas of artificial intelligence, for example, machine learning~\citep{ho1994decision,carterette2006learning,volkovs2014new,prasad2015distributional}, image processing~\citep{lumini2006detector,monwar2009multimodal}, information retrieval~\citep{li2014learning,aslam2001models,nuray2006automatic}, etc. The Maximin score is often used when the spread between the best and worst outcome is very large (see, e.g., p. 373 of \cite{mr91}). The maximin scoring method also has been used frequently in machine learning~\citep{wang2004feature,jiang2014diverse}, human computation~\citep{mao_hcomp12_2,Mao_aaai13}, etc.

\section{Preliminaries}\label{sec:prelim} 

We denote the disjoint union of sets by $\sqcup$. We denote the set of all permutations of a set $\mathcal{U}$ by $\mathcal{L(U)}$. For a positive integer $\ell$, we denote the set $\{1, \ldots, \ell\}$ by $[\ell]$. In most places, we ignore floors and ceilings for the sake of notational simplicity.
%Given a positive integer $t$, we denote the set $\{1, 2, \cdots, t\}$ by $[t]$. 

\subsection{Model of Input Data} The input data is an insertion-only stream of elements from some universe $\mathcal{U}$. In the context of voting, the input data is an insertion-only stream over the universe of all possible rankings (permutations). 

\subsection{Communication Complexity} We will use lower bounds on {\em communication complexity} of certain functions to prove space complexity lower bounds for our problems. Communication complexity of a function measures the number of bits that need to be exchanged between two players to compute a function whose input is split among those two players~\citep{yao1979some}. In a more restrictive {\it one-way communication model}, Alice, the first player, sends only one message to Bob, the second player, and Bob outputs the result. A protocol is a method that the players follow to compute certain functions of their input. Also the protocols can be randomized; in that case, the protocol needs to output correctly with probability at least $1-\delta$, for $\delta\in(0,1)$ (the probability is taken over the random coin tosses of the protocol). The randomized one-way communication complexity of a function $f$ with error probability $\delta$ is denoted by $\mathcal{R}_\delta^{\text{1-way}}(f)$. \citep{Kushilevitz} is a standard reference for communication complexity.

\subsection{Model of Computation} Our model of computation is the unit-cost RAM model on words of size $O(\log n)$, capable of generating uniformly random words and of performing arithmetic operations in $\{+, -, \log_2\}$ in one unit of time. We note that this model of computation has been used before \citep{dietzfelbinger1997reliable}. We store an integer $C$ using a variable length array of \citep{blandford2008compact} which allows us to read and update $C$ in $O(1)$ time and $O(\log C)$ bits of space.

\subsection{Universal Family of Hash Functions}
\begin{definition}(Universal family of hash functions)\\
 A family of functions $\mathcal{H} = \{ h | h:A\rightarrow B \}$ is called a {\em universal family of hash functions} if for all $a \ne b \in A$, $\Pr\{ h(a)= h(b) \} = \nfrac{1}{|B|}$, where $h$ is picked uniformly at random from $\mathcal{H}$.
 \longversion{\[ \Pr_{h \text{ picked uniformly at random from } \mathcal{H}}\{ h(a)= h(b) \} = \nfrac{1}{|B|} \]}
\end{definition}
We know that there exists a universal family of hash functions $\mathcal{H}$ from $[k]$ to $[\ell]$ for every positive integer $\ell$ and every prime $k$~\citep{cormen2001introduction}. Moreover, $|\mathcal{H}|$, the size of $\mathcal{H}$, is $O(k^2)$.

\longversion{
\subsection{Chernoff Bound}

We will use the following concentration inequality:

\begin{theorem}\label{thm:chernoff}
Let $X_1, \dots, X_\ell$ be a sequence of $\ell$ independent
random variables in $[0,1]$ (not necessarily identical). Let $S = \sum_i X_i$ and
let $\mu = \E{S}$. Then, for any $0 \leq \delta \leq 1$: 
$$\Pr[|S-\mu| \geq \delta \ell] < 2 \exp(-2\ell \delta^2)$$
and 
$$\Pr[|S - \mu| \geq \delta \mu] < 2\exp(-\delta^2\mu/3)$$
The first inequality is called an additive bound and the second
multiplicative. 
\end{theorem}
}

\subsection{Problem Definitions}

We now formally define the problems we study here. Suppose we have $0 < \eps < \varphi <1$.
\begin{definition}\textsc{$(\eps, \varphi)$-List heavy hitters}\\
 Given an insertion-only stream of length $m$ over a universe $\mathcal{U}$ of size $n$, find all items in $\mathcal{U}$ with frequency more than $\varphi m$, along with their frequencies up to an additive error of $\eps m$, and report no items with frequency less than $(\varphi-\eps)m$.
\end{definition}
 
\begin{definition}\textsc{$\eps$-Maximum}\\
 Given an insertion-only stream of length $m$ over a universe $\mathcal{U}$ of size $n$, find the maximum frequency up to an additive error of $\eps m$.
\end{definition}

Next we define the {\em minimum} problem for $0 < \eps <1$.

\begin{definition}\textsc{$\eps$-Minimum}\\
 Given an insertion-only stream of length $m$ over a universe $\mathcal{U}$ of size $n$, find the minimum frequency up to an additive error of $\eps m$.
\end{definition}

Next we define related heavy hitters problems in the context of rank aggregation. The input is a stream of rankings (permutations) over an item set $\UC$ for the problems below. The {\em Borda score} of 
an item $i$ is the sum, over all rankings, of the number of items $j \neq i$ for
which $i$ is ranked ahead of $j$ in the ranking.

\begin{definition}{\sc $(\eps, \varphi)$-List borda}\label{def:listborda}\\
 Given an insertion-only stream over a universe $\LC(\UC)$ where $|\UC|=n$, find all items with Borda score more than $\varphi mn$, along with their Borda score up to an additive error of $\eps mn$, and report no items with Borda score less than $(\varphi-\eps)mn$.
\end{definition}

\begin{definition}{\sc $\eps$-Borda}\label{def:borda}\\
 Given an insertion-only stream over a universe $\LC(\UC)$ where $|\UC|=n$, find the maximum Borda score up to an additive error of $\eps mn$.
\end{definition}

The {\em maximin score} of an item $i$ is the minimum, over all items $j \neq i$, of the number
of rankings for which $i$ is ranked ahead of $j$.

\begin{definition}{\sc $(\eps, \varphi)$-List maximin}\label{def:listmaximin}\\
 Given an insertion-only stream over a universe $\LC(\UC)$ where $|\UC|=n$, find all items with maximin score more than $\varphi m$ along with their maximin score up to an additive error of $\eps m$, and report no items with maximin score less than $(\varphi-\eps)m$.
\end{definition}

\begin{definition}{\sc $\eps$-maximin}\label{def:maximin}\\
 Given an insertion-only stream over a universe $\LC(\UC)$ where $|\UC|=n$, find the maximum maximin score up to an additive error of $\eps m$.
\end{definition}

Notice that the maximum possible Borda score of an item is $m(n-1) = \Theta(mn)$ and the maximum possible maximin score of an item is $m$. This justifies the approximation factors in \Cref{def:listborda,def:borda,def:listmaximin,def:maximin}. We note that finding an item with maximum Borda score within additive $\eps m n$ or maximum maximin score within additive $\eps m$ corresponds to finding an approximate winner of an election (more precisely, what is known as an $\eps$-winner)~\citep{deysampling}.

\section{Algorithms}

In this section, we present our upper bound results. All omitted proofs are in Appendix B. 
Before describing specific algorithms, we record some claims for later use.\longversion{ We begin with the following space efficient algorithm for picking an item uniformly at random from a universe of size $n$ below.}\shortversion{ \Cref{lem:sampling_ub} follows by checking whether we get all heads in $\log m$ tosses of a fair coin.}

\begin{restatable}{lemma}{ObsSamplingUB}\label{lem:sampling_ub}
 Suppose $m$ is a power of two\footnote{In all our algorithms, whenever we pick an item with probability $p>0$, we can assume, without loss of generality, that $\nfrac{1}{p}$ is a power of two. If not, then we replace $p$ with $p'$ where $\nfrac{1}{p'}$ is the largest power of two less than $\nfrac{1}{p}$. This does not affect correctness and performance of our algorithms.}. Then there is an algorithm $\mathcal{A}$ for choosing an item with probability $\nfrac{1}{m}$ that has space complexity of  $O(\log\log m)$ bits and time complexity of $O(1)$ in the unit-cost RAM model.
\end{restatable}
\begin{proof}
 We generate a $(\log_2 m)$-bit integer $C$ uniformly at random and record the sum of the digits in $C$. Choose the item only if the sum of the digits is $0$, i.e. if $C=0$.
\end{proof}
We remark that the algorithm in \Cref{lem:sampling_ub} has optimal space complexity as shown in \Cref{lem:sampling_lb} in Appendix B.

\longversion{
We remark that the algorithm in \Cref{lem:sampling_ub} has optimal space complexity as shown in \Cref{lem:sampling_lb} below which may be of independent interest.\longversion{ We also note that every algorithm needs to toss a fair coin at least $\Omega(\log m)$ times to perform any task with probability at least $\nfrac{1}{m}$.}

\begin{restatable}{proposition}{PropSamplingLB}[$\star$]\label{lem:sampling_lb}
 Any algorithm that chooses an item from a set of size $n$ with probability $p$ for $0< p \le \frac{1}{n}$, in unit cost RAM model must use $\Omega(\log\log m)$ bits of memory. 
\end{restatable}
}
Our second claim is a standard result for universal families of hash functions.
\begin{restatable}{lemma}{LemUnivHash}\label{lem:hash}
 For $S\subseteq A$, $\delta \in (0,1)$, and universal family of hash functions $\mathcal{H} = \{ h | h: A \rightarrow [\lceil\nfrac{ |S|^2 }{\delta}\rceil]\}$:
 \[ \Pr_{{h \in_U \mathcal{H}}} [\exists i\ne j\in S, h(i) = h(j) ] \le \delta \]
\end{restatable}
\begin{proof}
 For every $i\ne j \in S$, since $\mathcal{H}$ is a universal family of hash functions, we have
$\Pr_{h \in_\text{U} \mathcal{H}} [h(i) = h(j)] \le \frac{1}{\lceil{ |S|^2 }/{\delta}\rceil}$.
 Now we apply the union bound to get
$ \Pr_{h \in_{{U}}\mathcal{H}} [\exists i\ne j\in S, h(i) = h(j)] \le \frac{|S|^2}{\lceil{ |S|^2 }/{\delta}\rceil} \le \delta$
 \end{proof}

Our third claim is folklore and also follows from the celebrated DKW inequality \cite{DKW}. 
We provide a simple proof here that works for constant $\delta$.
\begin{restatable}{lemma}{LemDKW}\label{lem:dkw}
 Let $f_i$ and $\hat{f}_i$ be the frequencies of an item $i$ in a stream $\SC$ and in a random sample $\TC$  of size $r$ from $\SC$, respectively. Then for $r \geq  2\eps^{-2}\log(2\delta^{-1})$, with probability $1-\delta$, 
for every universe item $i$ simultaneously,  
$$\left |\frac{\hat{f}_i}{r} - \frac{f_i}{m} \right | \leq \epsilon.$$
\end{restatable}
\begin{proof}[Proof for constant $\delta$]
 This follows by Chebyshev's inequality and a union bound. Indeed, consider a given $i \in [n]$ with frequency $f_i$ and suppose
we sample each of its occurrences pairwise-independently with probability $r/m$, for a parameter $r$.  
Then the expected number ${\bf E}[\hat{f_i}]$ of sampled occurrences is 
$f_i \cdot r/m$ and the variance ${\bf Var}[\hat{f_i}]$ is $f_i \cdot r/m (1-r/m) \leq f_ir/m$. Applying Chebyshev's inequality, 
$$
\Pr \left [\left | \hat{f_i} - {\bf E}[\hat{f_i}] \right | \geq \frac{r \epsilon}{2} \right ] \leq \frac{{\bf Var}[\hat{f_i}]}{(r \epsilon/2)^2}
\leq \frac{4f_ir}{m r^2 \epsilon^2}.$$
Setting $r = \frac{C}{\epsilon^2}$ for a constant $C > 0$ makes this probability at most $\frac{4f_i}{C m}$. By the union bound, if we
sample each element in the stream independently with probability $\frac{r}{m}$, then the probability there exists an $i$ for which
$|\hat{f_i} - {\bf E}[\hat{f_i}]| \geq \frac{r \epsilon}{2}$ is at most $\sum_{i = 1}^n \frac{4 f_i}{C m} \leq \frac{4}{C}$, which for
$C \geq 400$ is at most $\frac{1}{100}$, as desired. 
\end{proof}

For now, assume that the length of the stream is known in advance; we show in \cref{sec:unknown} how to remove this assumption.

\subsection{List Heavy Hitters}\label{sec:lhh}
For the \lhh problem, we present two algorithms. The first is slightly suboptimal, but simple conceptually and already constitutes a very large improvement in the space complexity over known algorithms. We expect that this algorithm could be useful in practice as well. The second algorithm is more complicated, building on ideas from the first algorithm, and achieves the optimal space complexity upto constant factors.

We note that both algorithms proceed by sampling $O(\epsilon^{-2} \ln (1/\delta))$ stream items and updating a data structure as the stream progresses. In
both cases, the time to update the data structure is bounded
by $O(1/\epsilon)$, and so, under the standard assumption that the
length of the stream is at least
$\textrm{poly}(\ln(1/\delta) \epsilon)$, 
the time to perform this update can be
spread out across the next $O(1/\epsilon)$ stream updates, 
since with large probability there will be no items sampled among these 
next $O(1/\epsilon)$ stream updates. Therefore, we achieve worst-case\footnote{We emphasize that this is stronger than an amortized guarantee, as on every insertion, the cost will be $O(1)$.}
update time of $O(1)$.

\subsubsection{A simpler, near-optimal algorithm}

\begin{restatable}{theorem}{ThmMisra}\label{thm:heavy_hitters}
 Assume the stream length is known beforehand. Then there is a randomized one-pass algorithm $\mathcal{A}$ for the \textsc{$(\eps, \varphi)$-List heavy hitters} problem which succeeds with probability at least $1-\delta$ using $O\left(\eps^{-1}(\log\eps^{-1} + \log\log\delta^{-1}) + \varphi^{-1}\log n + \log\log m \right)$ bits of space. Moreover, $\mathcal{A}$ has an update time of $O(1)$ and reporting time linear in its output size. 
\end{restatable}

\paragraph{Overview}The overall idea is as follows. We sample $\el=O(\eps^{-2})$ many items from the stream uniformly at random as well as hash the id's (the word ``id'' is short for identifier) of the sampled elements into a space of size $O(\eps^{-4})$. Now, both the stream length as well as the universe size are $\text{poly}(\eps^{-1})$. From \cref{lem:dkw}, it suffices to solve the heavy hitters problem on the sampled stream. From \cref{lem:hash}, because the hash function is chosen from a universal family, the sampled elements have distinct hashed id's. We can then feed these elements into a standard Misra-Gries data structure with $\eps^{-1}$ counters, incurring space $O(\eps^{-1} \log \eps^{-1})$. Because we want to return the unhashed element id's for the heavy hitters, we also use $\log n$ space for recording the $\phi^{-1}$ top items according to the Misra-Gries data structure and output these when asked to report.

\begin{algorithm}[!h]
  \caption{for \textsc{$(\eps, \varphi)$-List heavy hitters}
  \label{alg:heavy_hitters}}
  \begin{algorithmic}[1]
    \Require{A stream $\SC$ of length $m$ over $\UC = [n]$; let $f(x)$ be the frequency of $x\in \UC$ in $\SC$}    
    \Ensure{A set $X\subseteq \UC$ and a function $\hat{f}:X \rightarrow \NB$ such that if $f(x)\ge \varphi m$, then $x\in X$ and $f(x) -\eps m \le \hat{f}(x) \le f(x) + \eps m$ and if $f(y) \le (\phi-\eps)m$, then $y\notin X$ for every $x, y \in \UC$}
    \Initialize{\vspace{-2ex}
     \Let{$\ell$}{$\nfrac{6\log(\nfrac{6}{\delta})}{\eps^2}$}\label{alg:ell}
     \State Hash function $h$ uniformly at random from a universal family $\HC \subseteq \{h:[n]\rightarrow \lceil \nfrac{4\ell^2}{\delta} \rceil\}$.
     \State An empty table $\TC_1$ of (key, value) pairs of length ${\eps^{-1}}$. Each key entry of $\TC_1$ can store an integer in $[0, \lceil \nfrac{400\ell^2}{\delta} \rceil]$ and each value entry can store an integer in $[0, 11\ell]$.\label{alg:m_table} \Comment{The table $\TC_1$ will be in sorted order 	by value throughout.}
     \State An empty table $\TC_2$ of length $\nfrac{1}{\varphi}$. Each entry of $\TC_2$ can store an integer in $[0, n]$. \Comment{The entries of $\TC_2$ will correspond to ids of the keys in $\TC_1$ of the highest $\nfrac{1}{\varphi}$ values}
    }
~\\
\Procedure{Insert}{x}
            \State With probability $p = \nfrac{6\ell}{m}$, continue. Otherwise, \Return.
       \State Perform Misra-Gries update using $h(x)$ maintaining $\TC_1$ sorted by values.\label{alg:u1}
       \If{The value of $h(x)$ is among the highest $\nfrac{1}{\varphi}$ valued items in $\TC_1$}
        \If{$x_i$ is not in $\TC_2$}
         \If{$\TC_2$ currently contains $\nfrac{1}{\varphi}$ many items}
         \State For $y$ in $\TC_2$  such that $h(y)$ is not among the highest $\nfrac{1}{\varphi}$ valued items in $\TC_1$, replace $y$ with $x$. 
         \Else
         \State We put $x$ in $\TC_2$.
         \EndIf
        \EndIf
        \State	Ensure that elements in $\TC_2$ are ordered according to corresponding values in $\TC_1$.
       \EndIf\label{alg:u2}
\EndProcedure
~\\
\Procedure{Report}{~}
    \State \Return items in $\TC_2$ along with their corresponding values in $\TC_1$ 
    \EndProcedure
  \end{algorithmic}
\end{algorithm}
\begin{proof}[Proof of \cref{thm:heavy_hitters}]
 The pseudocode of our \textsc{$(\eps, \varphi)$-List heavy hitters} algorithm is in \Cref{alg:heavy_hitters}. By \cref{lem:dkw}, if we select a subset $\mathcal{S}$ of size at least $\ell = 6\eps^{-2}{\log({6}\delta^{-1})}$ uniformly at random from the stream, then $\Pr[\forall i\in \mathcal{U}, |(\nfrac{\hat{f_i}}{|\SC|}) - (\nfrac{f_i}{n})| \le \eps] \ge 1 - \nfrac{\delta}{3}$, where $f_i$ and $\hat{f_i}$ are the frequencies of item $i$ in the input stream and $\mathcal{S}$ respectively. First we show that with the choice of $p$ in line \ref{alg:p} in \Cref{alg:heavy_hitters}, the number of items sampled is at least $\ell$ and at most $11\ell$ with probability at least $(1-\nfrac{\delta}{3})$. Let $X_i$ be the indicator random variable of the event that the item $x_i$ is sampled for $i\in[m]$. Then the total number of items sampled $X = \sum_{i=1}^m X_i$. We have $\EB[X] = 6\ell$ since $p = \nfrac{6\ell}{m}$. Now we have the following.
 \[ \Pr[X \le \ell \text{ or } X\ge 11\ell] \le \Pr[|X-\EB[X]|\ge 5\ell] \le \nfrac{\delta}{3} \]
{ The inequality follows from the Chernoff bound and the value of $\ell$. }From here onwards we assume that the number of items sampled is in $[\ell, 11\ell]$.
 
 We use (a modified version of) the Misra-Gries algorithm~\citep{misra82} to estimate the frequencies of items in $\mathcal{S}$. The length of the table in the Misra-Gries algorithm is $\eps^{-1}$. We pick a hash function $h$ uniformly at random from a universal family $\mathcal{H} = \{h | h:[n]\rightarrow \lceil\nfrac{4\ell^2}{\delta}\rceil\}$  of hash functions of size $|\HC| = O(n^2)$. Note that picking a hash function $h$ uniformly at random from $\mathcal{H}$ can be done using $O(\log n)$ bits of space.  
%$O(\log\log n + \log\log\nfrac{1}{\delta} + \log\log\nfrac{1}{\eps})$ bits of space by \Cref{lem:sampling_ub} and the value of $\ell$. 
\Cref{lem:hash} shows that there are no collisions in $\mathcal{S}$ under this hash function $h$ with probability at least $1-{\delta}/{3}$. From here onwards we assume that there is no collision among the ids of the sampled items under the hash function $h$.
 
 We modify the Misra-Gries algorithm as follows. Instead of storing the id of any item $x$ in the Misra-Gries table (table $\TC_1$ in line \ref{alg:m_table} in \Cref{alg:heavy_hitters}) we only store the hash $h(x)$ of the id $x$. We also store the ids (not the hash of the id) of the items with highest $\nfrac{1}{\varphi}$ values in $\TC_1$ in another table $\TC_2$. Moreover, we always maintain the table $\TC_2$ consistent with the table $\TC_1$ in the sense that the $i^{th}$ highest valued key in $\TC_1$ is the hash of the $i^{th}$ id in $\TC_2$. 
 %We maintain $\TC_1$ and $\TC_2$ using a data structure based on the differential encoding scheme of Theorem $2$ of~\citep{demaine2002frequency}. This data structure allows us to perform updates (in lines \ref{alg:u1} to \ref{alg:u2} in \Cref{alg:heavy_hitters}) in $O(1)$ time as explained below.
 
 Upon picking an item $x$ with probability $p$, we create an entry corresponding to $h(x)$ in $\TC_1$ and make its value one if there is space available in $\TC_1$; decrement the value of every item in $\TC_1$ by one if the table is already full; increment the entry in the table corresponding to $h(x)$ if $h(x)$ is already present in the table. When we decrement the value of every item in $\TC_1$, the table $\TC_2$ remains consistent and we do not need to do anything else. Otherwise there are three cases to consider. Case $1$: $h(x)$ is not among the $\nfrac{1}{\varphi}$ highest valued items in $\TC_1$. In this case, we do not need to do anything else. Case $2$: $h(x)$ was not among the $\nfrac{1}{\varphi}$ highest valued items in $\TC_1$ but now it is among the $\nfrac{1}{\varphi}$ highest valued items in $\TC_1$. In this case the last item $y$ in $\TC_2$ is no longer among the $\nfrac{1}{\varphi}$ highest valued items in $\TC_1$. We replace $y$ with $x$ in $\TC_2$. Case 3: $h(x)$ was among the $\nfrac{1}{\varphi}$ highest valued items in $\TC_1$. 
 %In this case, we update $\TC_1$ and $\TC_2$. Theorem $2$ of~\citep{demaine2002frequency} implies that every operation in cases $1-3$ can be performed in $O(1)$ time. 
 When the stream finishes, we output the ids of all the items in table $\TC_2$ along with the values corresponding to them in table $\TC_1$. Correctness follows from the correctness of the Misra-Gries algorithm and the fact that there is no collision among the ids of the sampled items.
\end{proof}

\subsubsection{An optimal algorithm}

\begin{restatable}{theorem}{ThmOpt}\label{thm:heavy_hitters2}
 Assume the stream length is known beforehand. Then there is a randomized one-pass algorithm $\mathcal{A}$ for the \textsc{$(\eps, \varphi)$-List heavy hitters} problem which succeeds with constant probability using $O\left(\eps^{-1}\log \phi^{-1}  + \phi^{-1}\log n + \log\log m \right)$ bits of space. Moreover,  $\mathcal{A}$ has an update time of $O(1)$ and reporting time linear in its output size.
\end{restatable}

Note that in this section, for the sake of simplicity, we ignore floors and ceilings and state the results for a constant error probability, omitting the explicit dependence on $\delta$.

\begin{algorithm}[!ht]
  \caption{for \textsc{$(\eps, \varphi)$-List heavy hitters}
  \label{alg:heavy_hitters2}}
  \begin{algorithmic}[1]
    \Require{A stream $\SC$ of length $m$ over universe $\UC = [n]$; let $f(x)$ be the frequency of $x\in \UC$ in $\SC$}    
    \Ensure{A set $X\subseteq \UC$ and a function $\hat{f}:X \rightarrow \NB$ such that if $f(x)\ge \varphi m$, then $x\in X$ and $f(x) -\eps m \le \hat{f}(x) \le f(x) + \eps m$ and if $f(y) \le (\phi-\eps)m$, then $y\notin X$ for every $x, y \in \UC$}
    \Initialize{\vspace{-2ex}
     \Let{$\ell$}{$10^5  \eps^{-2}$}\label{alg:ell}
     \Let{$s$}{0}
     \State Hash functions $h_1, \dots, h_{200\log(12\phi^{-1})}$ uniformly at random from a universal family $\HC \subseteq \{h:[n]\to [\nfrac{100}{\eps}]\}$.
     \State An empty table $\TC_1$ of (key, value) pairs of length ${2}{\phi^{-1}}$. Each key entry of $\TC_1$ can store an element of $[n]$ and each value entry can store an integer in $[0, 10\ell]$.\label{alg:m_table}
     \State An empty table $\TC_2$ with $100 \eps^{-1}$ rows and $200 \log (12\phi^{-1})$ columns. Each entry of $\TC_2$ can store an integer in $[0, 100\eps \ell]$.
     \State An empty $3$-dimensional table $\TC_3$ of size at most $100 \eps^{-1} \times 200 \log(12\phi^{-1}) \times 4 \log(\eps^{-1})$. Each entry of $\TC_3$ can store an  integer in $[0, 10 \ell]$. \Comment These are upper bounds; not all the allowed cells will actually be used. 
    }
    ~\\
    \Procedure{Insert}{$x$}
      \State With probability $\nfrac{\ell}{m}$, increment $s$ and continue. Else, \Return\label{alg:p}
      \State Perform Misra-Gries update on $\TC_1$ with $x$.\label{alg:u1}
      \For{$j \gets 1 \text{ to } 200\log(12\phi^{-1})$}   
      	        \State $i \gets h_j(x)$
	   \State With probability $\eps$, increment $\TC_2[i,j]$
      \State $t \gets \lfloor \log(10^{-6} \TC_2[i,j]^{2})\rfloor$ and $p \gets \min(\eps \cdot 2^{t}, 1)$\label{alg:tpcomp} 
	      \If{$t \geq 0$}
	      \State With probability $p$, increment $\TC_3[i, j, t]$ \label{alg:t3inc}
	      \EndIf 
      \EndFor
      \EndProcedure
 ~\\   
    \Procedure{Report}{~}
    \State $X \gets \emptyset$
    \For{{each key} $x \text{ with nonzero value in } \TC_1$}
    \For{$ j \gets 1 \text{ to } 200 \log(12\phi^{-1})$}
    \State $\hat{f}_j(x) \gets \sum_{t=0}^{4 \log(\eps^{-1})} \TC_3[h(x), j, t]/\min(\eps 2^{t}, 1)$ \label{lin:fj}
    \EndFor
    \State $\hat{f}(x) \gets \text{median}(\hat{f}_1, \dots, \hat{f}_{10 \log \phi^{-1}})$ \label{lin:fhat}
    \If{$\hat{f}(x) \geq (\phi -\eps/2)s$}
    \State $X \gets X \cup \{x\}$    
   \EndIf 
    \EndFor
    \State \Return $X, \hat{f}$
    \EndProcedure
  \end{algorithmic}
\end{algorithm}

\ignore{
Begin by observing that we can reduce the stream length to $\ell = O(\eps^{-2})$ by uniformly sampling $\ell$ elements from the stream. By \cref{lem:dkw}, with constant probability, it suffices to solve the $(\eps/2, \phi)$-\textsc{List heavy hitters} problem on the sampled stream. So, assume the stream length to be $\ell$ henceforth.
}
\paragraph{Overview}
As in the simpler algorithm, we sample $\ell = O(\eps^{-2})$ many stream elements and solve the $(\eps/2,\phi)$-\textsc{List heavy hitters} problem on this sampled stream. Also, the Misra-Gries algorithm for $(\phi/2, \phi)$-\textsc{List heavy hitters} returns a candidate set of $O(\phi^{-1})$ items containing all items of frequency at least $\phi \ell$. It remains to count the frequencies of these $O(\phi^{-1})$ items with upto $\eps \ell/2 = O(\eps^{-1})$ {\em additive} error, so that we can remove those whose frequency is less than $(\phi - \eps/2)\ell$. 

\ignore{Next, observe that we can efficiently find $O(\phi^{-1})$ items containing all items of frequency at least $\phi \ell$. Namely, we can run the (deterministic) Misra-Gries algorithm on the sampled stream to solve the $(\phi/2, \phi)$-\textsc{List heavy hitters problem}, using space $O(\phi^{-1} \log n)$ for $n \gg 1/\eps$.  Now, it remains to count the frequencies of these $O(\phi^{-1})$ items with upto $\eps \ell/2 = O(\eps^{-1)}$ additive error, so that we can remove those whose frequency is less than $(\phi - \eps/2)\ell$. }

Fix some item $i \in [n]$, and let $f_i$ be $i$'s count in the sampled stream. A natural approach to count $f_i$ approximately is to increment a 
 counter probabilistically, instead of deterministically, at every occurrence of $i$. Suppose that we increment a counter with probability $0\leq p_i \leq 1$ whenever item $i$ arrives in the stream. Let the value of the counter be $\hat{c}_i$, and let $\hat{f}_i = \hat{c}_i/p_i$. We see that $\E{\hat{f}_i} = f_i$ and $\mathsf{Var}[\hat{f}_i]  \leq f_i/p_i$. It follows that if $p_i = \Theta(\eps^2 f_i)$, then $\mathsf{Var}[\hat{f}_i] = O(\eps^{-2})$, and hence, $\hat{f}_i$ is an unbiased estimator of $f_i$ with additive error $O(\eps^{-1})$  with constant probability. We call such a counter an {\em accelerated counter} as the probability of incrementing accelerates with increasing counts. For each $i$, we can maintain $O(\log \phi^{-1})$ accelerated counters independently and take their median to drive the probability of deviating by more than $O(\eps^{-1})$ down to $O(\phi)$. So, with constant probability, the frequency for each of the $O(\phi^{-1})$ items in the Misra-Gries data structure is estimated within $O(\eps^{-1})$ error, as desired.

However, there are two immediate issues with this approach. The first problem is that we may need to keep counts for $\Omega(\ell) = \Omega(\eps^{-2})$ distinct items, which is too costly for our purposes. To get around this, we use a hash function from a universal family to hash the universe to a space of size $u = \Theta(\eps^{-1})$, and we work throughout with the hashed id's. We can then show that the space complexity for each iteration is $O(\eps^{-1})$. 
\ignore{Let the frequency of a hashed id be the sum of the frequencies of the id's preimage. Then, the number of hashed id's with frequency between $2^k/\eps$ and $2^{k+1}/\eps$ is $O((2^k \eps)^{-1})$, and the expected space used in the accelerated counter for such an id is at most $\log(\eps^2 (2^{k+1}/\eps)^2) = O(k)$. So the expected total space used for each iteration is $O(u) = O(\eps^{-1})$ for items with frequency less than $\eps^{-1}$ and $\sum_{k=0}^\infty \frac{O(k)}{2^k \eps} = O(\eps^{-1})$ for the rest, as desired.} Also, the accelerated counters now estimate frequencies of hashed id's instead of actual items, but because of universality, the expected frequency of any  hashed id is $\ell/u = O(\eps^{-1})$, our desired error bound.

The second issue is that we need a constant factor approximation of $f_i$, so that we can set $p_i$ to $\Theta	(\eps^2 f_i)$. But because  the algorithm needs to be one-pass, we cannot first compute $p_i$ in one pass and then run the accelerated counter in another. So, we divide the stream into {\em epochs} in which  $f_i$ stays within a factor of $2$, and use a different $p_i$ for each epoch. In particular, set $p_i^t = \eps \cdot 2^{t}$ for $0 \leq t \leq \log(p_i/\eps)$.  We want to keep a running estimate of $i$'s count to within a factor of $2$ to know if the current epoch should be incremented.
For this, we subsample each element of the stream with probability $\eps$ independently  and maintain exact counts for the observed hashed id's. It is easy to see that this requires only $O(\eps^{-1})$ bits in expectation. Consider any $i \in [u]$ and the prefix of the stream upto $b \leq \ell$, and let $f_i(b)$ be $i$'s frequency in the prefix, let $\bar{c}_i(b)$ be $i$'s frequency among the samples in the prefix, and $\bar{f}_i(b) = \frac{\bar{c}_i(b)}{\eps}$. We see that $\E{\bar{f}_i(b)} = f_i(b)$, 
\ignore{and $\mathsf{Var}[\bar{f}_i(b)] \leq \frac{f_i(b)}{\eps}$.  By Chebyshev, for any fixed $b$, $\Pr[|\bar{f}_i(b) - f_i(b)| > f_i(b)/\sqrt{2}] \leq \frac{2}{f_i(b) \eps}$, and hence, 
can show that $\bar{f}_i(b)$ is a $\sqrt{2}$-factor approximation of $f_i(b)$ with probability $1 - O((f_i(b) \eps)^{-1})$. Now, let $p_i(b) = \Theta(\eps^2 f_i(b))$, and for any epoch $t$, set $b_{i,t} = \min\{b: p_i(b) > p_i^{t-1}\}$. The last makes sense because $p_i(b)$ is non-decreasing with $b$. Also, note that $f_i(b_{i,t}) = \Omega(2^{t/2}/\eps)$. So, by the union bound, the probability that there exists $t$ for which $\bar{f}_i(b_{i,t})$ is not a $\sqrt2$-factor approximation of $f_i(b_{i,t})$ is at most $\sum_t \frac{1}{\Omega(f_i(b_{i,t}) \eps)} = \sum_t \frac{1}{\Omega(2^{t/2})}$, a small constant. In fact, it follows then that with constant probability, for all $b \in [\ell]$,  $\bar{f}_i(b)$ is a $2$-factor approximation of $f_i(b)$.}
Moreover, we show that for any $b \in [\ell]$, $\bar{f}_i(b)$ is a $4$-factor approximation of $f_i(b)$ with constant probability. By repeating $O(\log \phi^{-1})$ times independently and taking the median, the error probability can be driven down to $O(\phi)$. 

Now, for every hashed id $i \in [u]$, we need not one accelerated counter but $O(\log(\eps f_i))$ many, one corresponding to each epoch $t$. When an element with hash id $i$ arrives at position $b$, we decide, based on $\bar{f}_i(b)$, the epoch $t$ it belongs to and then increment the $t$'th accelerated counter with probability $p_i^t$.  The storage cost over all $i$ is still $O(1/\eps)$. Also, we iterate the whole set of accelerated counters $O(\log \phi^{-1})$ times, making the total storage cost $O(\eps^{-1}\log \phi^{-1})$.

Let $\hat{c}_{i,t}$ be the count in the accelerated counter for hash id $i$ and epoch $t$. Then, let $\hat{f}_i = \sum_t {\hat{c}_{i,t}}/{p_i^t}$. Clearly, $\E{\hat{f}_i} = f_i$. The variance is $O(\eps^{-2})$ in each epoch, and so, $\mathsf{Var}[\hat{f}_i]=O(\eps^{-2} \log \eps^{-1})$, not $O(\eps^{-2})$ which we wanted. This issue is fixed by a change in how the sampling probabilities are defined. We now go on to the formal proof.

\begin{proof}[Proof of \cref{thm:heavy_hitters2}]
Pseudocode appears in \cref{alg:heavy_hitters2}. Note that the numerical constants are chosen for convenience of analysis and have not been optimized. Also, for the sake of simplicity, the pseudocode does not have the optimal reporting time, but it can be modified to achieve this; see the end of this proof for details.

By standard Chernoff bounds, with probability at least $99/100$, the length of the sampled stream $\ell/10 \leq s \leq 10\ell$. For $x \in [n]$, let $f_{\text{samp}}(x)$ be the frequency of $x$ in the sampled stream. By \cref{lem:dkw}, with probability at least $9/10$, for all $x \in [n]$:
$$\left|\frac{f_{\text{samp}}(x)}{s} - \frac{f(x)}{m}\right| \leq \frac{\eps}{4}$$
Now, fix $j \in [10 \log \phi^{-1}]$ and $x \in [n]$. Let $i = h_j(x)$ and $f_i = \sum_{y: h_j(y) = h_j(x)} f_{\text{samp}}(y)$. Then, for a random $h_j \in \mathcal{H}$, the expected value of $\frac{f_i}{s} - \frac{f_{\text{samp}}(x)}{s}$ is $\frac{\eps}{100}$, since $\mathcal{H}$ is a universal mapping to a space of size $100\eps^{-1}$. Hence, using Markov's inequality and the above:
\begin{equation}\Pr\left[\left| \frac{f(x)}{m} - \frac{f_i}{s}\right| \geq \frac{\eps}{2}\right] 
\leq \Pr\left[\left| \frac{f(x)}{m} - \frac{f_{\text{samp}}}{s}\right| \geq \frac{\eps}{4}\right] + \Pr\left[\left| \frac{f_{\text{samp}}(x)}{m} - \frac{f_i}{s}\right| \geq \frac{\eps}{4}\right]
< \frac{1}{10}  + \frac{1}{25} < \frac{3}{20}\end{equation}
In Lemma \ref{lem:pmain} below, we show that for each $j \in [200 \log(12\phi^{-1})]$, with error probability at most $3/10$, $\hat{f}_j(x)$ (in line \ref{lin:fj}) estimates $f_i$ with additive error at most $5000\eps^{-1}$, hence estimating $\frac{f_i}{s}$ with additive error at most $\frac{\eps}{2}$. Taking the median over $200 \log( 12\phi^{-1})$ repetitions (line \ref{lin:fhat}) makes the error probability go down to $\frac{\phi}{6}$ using standard Chernoff bounds. Hence, by the union bound, with probability at least $2/3$, for each of the $2/\phi$ keys $x$ with nonzero values in $\TC_1$, we have an estimate of $\frac{f(x)}{m}$ within additive error $\eps$, thus showing correctness.

\begin{lemma}\label{lem:pmain}
Fix $x \in [n]$ and $j \in [200 \log 12\phi^{-1}]$, and let $i = h_j(x)$. Then, $\Pr[|\hat{f}_j(x) - f_i| > 5000 \eps^{-1}] \leq 3/10$, where $\hat{f}_j$ is the quantity computed in line \ref{lin:fj}.
\end{lemma}
\begin{proof}
Index the sampled stream elements $1, 2, \dots, s$, and for $b \in [s]$, let $f_i(b)$ be the frequency of items with hash id $i$ restricted to the first $b$ elements of the sampled stream. Let $\bar{f}_i(b)$ denote the value of $\TC_2[i,j]\cdot \eps^{-1}$ after the procedure \textsc{Insert} has been called for the first $b$ items of the sampled stream.
\begin{claim}\label{lem:approx}
With probability at least $9/10$, for all $b \in [s]$ such that $f_i(b) \geq 100 \eps^{-1}$, $\bar{f}_i(b)$ is within a factor of 4 of $f_i(b)$.
\end{claim}
\begin{proof}
Fix $b \in [s]$. Note that $\E{\bar{f}_i(b)} = f_i(b)$ as $\TC_2$ is incremented with rate $\eps$. $\mathsf{Var}[\bar{f}_i(b)] \leq f_i/\eps$, and so by Chebyshev's inequality:
$$\Pr[|\bar{f}_i(b) - f_i(b)| > f_i(b)/2] < \frac{4}{f_i(b) \eps}$$
We now break the stream into chunks, apply this inequality to each chunk and then take a union bound to conclude. Namely, for any integer $t\geq 0$, define $b_t$ to be the first	 $b$ such that $100 \eps^{-1} 2^{t} \leq f_i(b) < 100 \eps^{-1} 2^{t+1}$ if such a $b$ exists. Then:
\begin{align*}
\Pr[\exists t \geq 0: |\bar{f}_i(b_t) - f_i(b_t)| > f_i(b_t)/2] &< \sum_t \frac{4}{100 \cdot 2^{t-1}}\\ &< \frac{1}{10}
 \end{align*} 
 So, with probability at least $9/10$, every $\bar{f}_i(b_t)$ and $f_i(b_t)$ are within a factor of $2$ of each other. Since for every $b\geq b_0$, $f_i(b)$ is within a factor of ${2}$ from some $f_i(b_t)$, the claim follows.
\end{proof}
Assume the event in \cref{lem:approx} henceforth. Now, we are ready to analyze $\TC_3$ and in particular, $\hat{f}_j(x)$. First of all, observe that if $t<0$ in line \ref{alg:tpcomp}, at some position $b$ in the stream, then $\TC_2[i,j]$ at that time must be at most 1000, and so by standard Markov and Chernoff bounds,  with probability at least $0.85$, 
\begin{equation}\label{eqn:ass}
f_i(b) 
\begin{cases}
< 4000 \eps^{-1}, & \text{ if } t<0\\
> 100 \eps^{-1}, & \text { if } t \geq 0
\end{cases}
\end{equation}
Assume this event. Then, $f_i - 4000 \eps^{-1} \leq \E{\hat{f}_j(x)}  \leq f_i$.
\begin{claim}
$$\mathsf{Var}(\hat{f}_j(x)) \leq {20000}{\eps^{-2}}$$
\end{claim}
\begin{proof}
If the stream element at position $b$ causes an increment in $\TC_3$ with probability $\eps2^t$ (in line \ref{alg:t3inc}), then $1000 \cdot 2^{t/2} \leq \TC_2[i,j] \leq 1000\cdot 2^{(t+1)/2}$, and so, $\bar{f}_i(b) \leq 1000\eps^{-1} 2^{(t+1)/2}$.  This must be the case for the highest $b = \bar{b}_t$ at which the count for $i$  in $\TC_3$ increments at the $t$'th slot. The number of such occurrences of $i$ is at most $f_i(\bar{b}_t) \leq 4 \bar{f_i}(\bar{b}_t)\leq 4000 \eps^{-1}2^{(t+1)/2}$ by \cref{lem:approx} (which can be applied since  $f_i(b) > 100\eps^{-1}$ by Equation \ref{eqn:ass}). So:
\begin{align*}
\mathsf{Var}[\hat{f}_j(x)] \leq \sum_{t \geq 0} \frac{f_i(\bar{b}_t)}{\eps 2^t} \leq \sum_{t \geq 0} \frac{4000}{\eps^2} 2^{-t/3}\leq {20000}{\eps^{-2}}
\end{align*}
Elements inserted with probability $1$ obviously do not contribute to the variance.
\end{proof}
So, conditioning on the events mentioned, the probability that $\hat{f}_j(x)$ deviates from $f_i$ by more than $5000\eps^{-1}$ is at most $1/50$. Removing all the conditioning yields what we wanted:
$$\Pr[|\hat{f}_j(x) - f_i| > 5000\eps^{-1}] \leq \frac{1}{50} + \frac{3}{20} + \frac{1}{10} \leq 0.3$$ 
\end{proof}

We next bound the space complexity.
\begin{claim}\label{lem:spacehh}
With probability at least $2/3$, \cref{alg:heavy_hitters2} uses $O(\eps^{-1} \log \phi^{-1} + \phi^{-1} \log n + \log \log m)$ bits of storage, if $n = \omega(\eps^{-1})$.
\end{claim}
\begin{proof}
The expected length of the sampled stream is $\ell = O(\eps^{-2})$. So, the number of bits stored in $\TC_1$ is $O(\phi^{-1} \log n)$. For $\TC_2$, note that in lines 13-15, for any given $j$, $\TC_2$ is storing a total of $\eps \ell = O(\eps^{-1})$ elements in expectation. So, for $k \geq 0$, there can be at most $O((\eps 2^k)^{-1})$ hashed id's with counts between $2^k$ and $2^{k+1}$. Summing over all $k$'s and accounting for the empty cells gives $O(\eps^{-1})$ bits of storage, and so the total space requirement of $\TC_2$ is $O(\eps^{-1} \log \phi^{-1})$. .

The probability that a hashed id $i$ gets counted in table $\TC_3$ is at most $10^{-6}\eps^3 \bar{f}_i^2(s)$ from line \ref{alg:tpcomp} and our definition of $\bar{f}_i$ above. Moreover, from \cref{lem:approx}, we have that this is at most $16 \cdot 10^{-6} \eps^3 {f}_i^2(s)$  if $f_i > 100 \eps^{-1}$. Therefore, if $f_i = 2^k\cdot 100 \eps^{-1}$ with $k \geq 0$, then the expected value of a cell in $\TC_3$ with first coordinate $i$ is at most $1600\cdot 2^{2k} \eps = 2^{O(k)}$. Taking into account that there are at most $O((\eps 2^k)^{-1})$ many such id's $i$ and that the number of epochs $t$ associated with such an $i$ is at most $\log(16 \cdot 10^{-6} \eps^2 {f}_i^2) = O(\log(\eps f_i)) = O(k)$ (from line \ref{alg:tpcomp}), we get that the total space required for $\TC_3$ is:
\begin{align*}&\sum_{j = 1}^{O(\log \phi^{-1})} \left(O(\eps^{-1}) + \sum_{k=0}^\infty O((\eps 2^k)^{-1}) \cdot O(k) \cdot O(k)\right) \\ &= O(\eps^{-1} \log \phi^{-1})\end{align*}
where the first $O(\eps^{-1})$ term inside the summation is for the $i$'s with $f_i < 100 \eps^{-1}$. Since we have an expected space bound, we obtain a worst-case space bound with error probability $1/3$ by a Markov bound. 

The space required for sampling is an additional $O(\log \log m)$, using \cref{lem:sampling_ub}.
\end{proof}
We note that the space bound can be made worst case by aborting the algorithm if it tries to use more space.

The only remaining aspect of \cref{thm:heavy_hitters2} is the time complexity. As observed in \cref{sec:lhh}, the update time can be made $O(1)$ per insertion under the standard assumption of the stream being sufficiently long. The reporting time can also be made linear in the output by changing the bookkeeping a bit. Instead of computing $\hat{f}_j$ and $\hat{f}$ at reporting time, we can maintain them after every insertion. Although this apparently makes INSERT costlier, this is not true in fact because we can spread the cost over future stream insertions. The space complexity grows by a constant factor.
\end{proof}

\subsection{Maximum}

By tweaking \Cref{alg:heavy_hitters} slightly, we get the following result for the {\sc $\eps$-Maximum} problem. 

\begin{restatable}{theorem}{ThmHeavyHittersWeak}\label{thm:heavy_hitters_weak}
 Assume the length of the stream is known beforehand. Then there is a randomized one-pass algorithm $\mathcal{A}$ for the \textsc{$\eps$-Maximum} problem which succeeds with probability at least $1-\delta$ using $O\left(\min\{\nfrac{1}{\eps}, n\}(\log\nfrac{1}{\eps} + \log\log\nfrac{1}{\delta}) + \log n + \log\log m \right)$ bits of space. Moreover, the algorithm $\mathcal{A}$ has an update time of $O(1)$.
\end{restatable}

\begin{proof}
 Instead of maintaining the table $\TC_2$ in \Cref{alg:heavy_hitters}, we just store the actual id of the item with maximum frequency in the sampled items. 
\end{proof}

\subsection{Minimum}
\begin{restatable}{theorem}{ThmRare}\label{thm:rare}
 Assume the length of the stream is known beforehand. Then there is a randomized one-pass algorithm $\mathcal{A}$ for the \textsc{$\eps$-Minimum} problem which succeeds with probability at least $1-\delta$ using $O\left((\nfrac{1}{\eps})\log\log(\nfrac{1}{\eps\delta}) + \log\log m \right)$ bits of space. Moreover, the algorithm $\mathcal{A}$ has an update time of $O(1)$.
\end{restatable}

\paragraph{Overview}

Pseudocode is provided in \cref{alg:rare}.
The idea behind our $\epsilon$-Minimum problem is as follows. It is most easily explained by looking at the REPORT(x) procedure starting in line 13. In lines 14-15 we ask, is the universe size $|U|$ significantly larger than $1/\epsilon$? Note that if it is, then outputting a random item from $|U|$ is likely to be a solution. Otherwise $|U|$ is $O(1/\epsilon)$.

The next point is that if the number of distinct elements in the stream were smaller than $1/(\epsilon \log(1/\epsilon))$, then we could just store all the items together with their frequencies with $O(1/\epsilon)$ bits of space. Indeed, we can first sample $O(1/\epsilon^2)$ stream elements so that all relative frequencies are preserved up to additive $\epsilon$, thereby ensuring each frequency can be stored with $O(\log(1/\epsilon)$ bits. Also, since the universe size is $O(1/\epsilon)$, the item identifiers can also be stored with $O(\log(1/\epsilon)$ bits. So if this part of the algorithm starts taking up too much space, we stop, and we know the number of distinct elements is at least $1/(\epsilon \log(1/\epsilon))$, which means that the minimum frequency is at most $O(m \epsilon \log(1/\epsilon))$. This is what is being implemented in steps 9-10 and 18-19 in the algorithm.

We can also ensure the minimum frequency is at least $\Omega(m \epsilon / \log(1/\epsilon))$. Indeed, by randomly sampling $O((\log(1/\epsilon)/\epsilon)$ stream elements, and maintaining a bit vector for whether or not each item in the universe occurs - which we can with $O(1/\epsilon)$ bits of space since $|U| = O(1/\epsilon)$ - any item with frequency at least $\Omega(\epsilon m/ \log(1/\epsilon))$ will be sampled and so if there is an entry in the bit vector which is empty, then we can just output that as our solution. This is what is being implemented in steps 8 and 16-17 of the algorithm.

Finally, we now know that the minimum frequency is at least $\Omega(m \epsilon / \log(1/\epsilon))$ and at most $O(m \epsilon \log(1/\epsilon))$. At this point if we randomly sample $O((\log^6 1/\epsilon)/\epsilon)$ stream elements, then by Chernoff bounds all item frequencies are preserved up to a relative error factor of $(1 \pm 1/\log^2 (1/\epsilon))$, and in particular the relative minimum frequency is guaranteed to be preserved up to an additive $\epsilon$. At this point we just maintain the exact counts in the sampled stream but truncate them once they exceed $\textrm{poly}(\log(1/\epsilon)))$ bits, since we know such counts do not correspond to the minimum. Thus we only need $O(\log \log (1/\epsilon))$ bits to represent their counts. This is implemented in step 11 and step 20 of the algorithm. 

\begin{algorithm}[!t]
  \caption{for \textsc{$\eps$-Minimum}
    \label{alg:rare}}
  \begin{algorithmic}[1]
    \Require{A stream $\SC = (x_i)_{i\in[m]}\in \UC^m$ of length $m$ over $\UC$; let $f(x)$ be the frequency of $x\in \UC$ in $\SC$}    
    \Ensure{An item $x\in\UC$ such that $f(x) \le f(y) + \eps m$ for every $y\in\UC$}
    \Initialize{
    \Let{$\ell_1$}{$\nfrac{\log(\nfrac{6}{\eps\delta})}{\eps}$}, $\ell_2 \leftarrow \nfrac{\log(\nfrac{6}{\delta})}{\eps^2}$, $\ell_3 \leftarrow \nfrac{\log^6 (\nfrac{6}{\delta\eps})}{\eps} $
    \Let{$p_1$}{$\nfrac{6\ell_1}{m}$}, $p_2 \leftarrow \nfrac{6\ell_2}{m}$, $p_3 \leftarrow \nfrac{6\ell_3}{m}$
    \Let{$\SC_1, \SC_2, \SC_3$}{$\emptyset$}
    \Let{$\BC_1$}{the bit vector for $\SC_1$}}
    ~\\
    \Procedure{Insert}{x}
      \State Put $x$ in $\SC_1$ with probability $p_1$ by updating the bit vector $\BC_1$
      \If{the number of distinct items in the stream so far is at most $\nfrac{1}{(\eps\log(\nfrac{1}{\eps}))}$}
	\State Pick $x$ with probability $p_2$ and put the id of $x$ in $\SC_2$ and initialize the corresponding counter to $1$ if $x\notin\SC_2$ and increment the counter corresponding to $x$ by $1$.
      \EndIf
      \State Pick $x$ with probability $p_3$, put the id of $x$ in $\SC_3$ and initialize the corresponding counter to $1$ if $x_i\notin\SC_3$ and increment the counter corresponding to $x_i$ by $1$. Truncate counters of $\SC_3$ at $2\log^7(\nfrac{2}{\eps\delta})$.
    \EndProcedure
    ~\\
    
    \Procedure{Report}{~}
    \If{$|\UC|\ge \nfrac{1}{((1-\delta)\eps)}$}
      \State \Return an item $x$ from the first $\nfrac{1}{((1-\delta)\eps)}$ items in $\UC$ (ordered arbitrarily) uniformly at random
      \EndIf
   \If{$\SC_1 \ne \UC$}
      \State \Return{any item from $\UC\setminus\SC_1$}\label{alg:S1_out}
    \EndIf
    \If{the number of distinct items in the stream is at most $\nfrac{1}{(\eps\log(\nfrac{1}{\eps}))}$}
      \State \Return an item in $\SC_2$ with minimum counter value in $\SC_2$\label{alg:S2_out}
    \EndIf
    \State \Return the item with minimum frequency in $\SC_3$
    \EndProcedure
  \end{algorithmic}
\end{algorithm}

\begin{proof}[Proof of \cref{thm:rare}]
 The pseudocode of our \textsc{$\eps$-Minimum} algorithm is in \Cref{alg:rare}. If the size of the universe $|\UC|$ is at least $\nfrac{1}{((1-\delta)\eps)}$, then we return an item $x$ chosen from $\UC$ uniformly at random. Note that there can be at most $\nfrac{1}{\eps}$ many items with frequency at least $\eps m$. Hence every item $x$ among other remaining $\nfrac{\delta}{((1-\delta)\eps)}$ many items has frequency less than $\eps m$ and thus is a correct output of the instance. Thus the probability that we answer correctly is at least $(1-\delta)$. From here on, let us assume $|\UC|<\nfrac{1}{((1-\delta)\eps)}$.
 
 Now, by the value of $p_j$, it follows from the proof of \Cref{thm:heavy_hitters} that we can assume $\ell_j < |\SC_j| < 11\ell_j$ for $j = 1, 2, 3$ which happens with probability at least $(1-(\nfrac{\delta}{3}))$. We first show that every item in $\UC$ with frequency at least $\eps m$ is sampled in $\SC_1$ with probability at least $(1-(\nfrac{\delta}{6}))$. For that, let $X_i^j$ be the indicator random variable for the event that the $j^{th}$ sample in $\SC_1$ is item $i$ where $i\in\UC$ is an item with frequency at least $\eps m$. Let $\HC\subset\UC$ be the set of items with frequencies at least $\eps m$. Then we have the following.
 \longversion{
 \[ \Pr[X_i^j = 0] = 1-\eps \Rightarrow \Pr[X_i^j = 0 ~\forall j\in\SC_1] \le (1-\eps)^{\ell_1} \le \exp\{-\eps\ell_1\} = \nfrac{\eps\delta}{6} \]
 }
 \shortversion{
 \begin{eqnarray*}
  &&\Pr[X_i^j = 0] = 1-\eps \\
  &\Rightarrow& \Pr[X_i^j = 0 ~\forall j\in\SC_1] \le (1-\eps)^{\ell_1} \le \exp\{-\eps\ell_1\} = \nfrac{\eps\delta}{6}
 \end{eqnarray*}
 }
 Now applying union bound we get the following.
 \[ \Pr[ \exists i\in \HC, X_i^j = 0 ~\forall j\in\SC_1] \le (\nfrac{1}{\eps})\nfrac{\eps\delta}{6} \le \nfrac{\delta}{6} \]
 Hence with probability at least $(1-(\nfrac{\delta}{3})-(\nfrac{\delta}{6})) \ge (1-\delta)$, the output at line \ref{alg:S1_out} is correct. Now we show below that if the frequency of any item $x\in\UC$ is at most $\nfrac{\eps\ln(\nfrac{6}{\delta})}{\ln(\nfrac{6}{\eps \delta})}$, then $x\in \SC_1$ with probability at least $(1-(\nfrac{\delta}{6}))$.
 \[ \Pr[ x\notin \SC_1 ] = (1-\nfrac{\eps\ln (\nfrac{6}{\delta})}{\ln (\nfrac{6}\eps\delta)})^{\nfrac{\ln(\nfrac{6}{\eps\delta})}{\eps}} \le \nfrac{\delta}{6} \]
 Hence from here onwards we assume that the frequency of every item in $\UC$ is at least $\nfrac{\eps m\ln (\nfrac{6}{\delta})}{\ln (\nfrac{6}\eps\delta)}$. 
 
 If the number of distinct elements is at most $\nfrac{1}{(\eps\ln(\nfrac{1}{\eps}))}$, then line \ref{alg:S2_out} outputs the minimum frequency item up to an additive factor of $\eps m$ due to Chernoff bound. Note that we need only $O(\ln(\nfrac{1}{((1-\delta)\eps)}))$ bits of space for storing ids. Hence $\SC_2$ can be stored in space $O((\nfrac{1}{\eps\ln(\nfrac{1}{\eps})}) \ln(\nfrac{1}{((1-\delta)\eps)} \ln\ln(\nfrac{1}{\delta})) = O(\nfrac{1}{\eps} \ln\ln(\nfrac{1}{\delta}))$.
 
 Now we can assume that the number of distinct elements is at least $\nfrac{1}{(\eps\ln(\nfrac{1}{\eps}))}$. Hence if $f(t)$ is the frequency of the item $t$ with minimum frequency, then we have $ m\nfrac{\eps}{\ln (\nfrac{1}{\eps})} \le f(t) \le m\eps \ln(\nfrac{1}{\eps})$. 
 
 Let $f_i$ be the frequency of item $i\in \UC$, $e_i$ be the counter value of $i$ in $\SC_3$, and $\hat{f}_i = \nfrac{e_i m}{\ell_3}$. Now again by applying Chernoff bound we have the following for any fixed $i\in\UC$.
\begin{eqnarray*}
\Pr[ |f_i - \hat{f}_i| > \nfrac{f_i}{\ln^2(\nfrac{1}{\eps})} ] 
&\le & 2\exp\{ -\nfrac{\ell_3 f_i}{(m \ln^4(\nfrac{1}{\eps}))} \}\\
& \le & 2 \exp\{ -\nfrac{f_i\ln^2 (\nfrac{6}{\eps\delta})}{(\eps m)} \}\\
& \le & \nfrac{\eps\delta}{6}.
\end{eqnarray*}
 Now applying union bound we get the following using the fact that $|\UC|\le \nfrac{1}{\eps (1-\delta)}$.
 \[ \Pr[ \forall i\in\UC, |f_i - \hat{f}_i| \le \nfrac{f_i}{\ln^2(\nfrac{1}{\eps})} ] > 1-\nfrac{\delta}{6} \]
 Again by applying Chernoff bound and union bound we get the following.
 \[ \Pr[ \forall i\in\UC \text{ with } f_i > 2m\eps\ln(\nfrac{1}{\eps}), |f_i - \hat{f}_i| \le \nfrac{f_i}{2} ] > 1-\nfrac{\delta}{6} \]
 Hence the items with frequency more than $2m\eps\ln(\nfrac{1}{\eps})$ are approximated up to a multiplicative factor of $\nfrac{1}{2}$ from below in $\SC_3$. The counters of these items may be truncated. The other items with frequency at most $2m\eps\ln(\nfrac{1}{\eps})$ are be approximated up to $(1 \pm \nfrac{1}{\ln^2(\nfrac{1}{\eps})})$ relative error and thus up to an additive error of $\nfrac{\eps m}{3}$. The counters of these items would not get truncated. Hence the item with minimum counter value in $\SC_3$ is the item with minimum frequency up to an additive $\eps m$.
 
 We need $O(\ln(\nfrac{1}{\eps\delta}))$ bits of space for the bit vector $\BC_1$ for the set $\SC_1$. We need $O(\ln^2(\nfrac{1}{\eps\delta}))$ bits of space for the set $\SC_2$ and $O((\nfrac{1}{\eps})\ln\ln(\nfrac{1}{\eps\delta}))$ bits of space for the set $\SC_3$ (by the choice of truncation threshold). We need an additional $O\left( \ln\ln m \right)$ bits of space for sampling using \Cref{lem:sampling_ub}. Moreover, using the data structure of Section 3.3 of \citep{demaine2002frequency} \Cref{alg:rare} can be performed in $O(1)$ time. Alternatively, we may also use the strategy described in \cref{sec:lhh} of spreading update operations over several insertions to make the cost per insertion be $O(1)$.
\end{proof}

\subsection{Borda and Maximin}
\begin{restatable}{theorem}{ThmBorda}\label{thm:borda}
 Assume the length of the stream is known beforehand. Then there is a randomized one-pass algorithm $\mathcal{A}$ for \textsc{$(\eps, \varphi)$-List Borda} problem which succeeds with probability at least $1-\delta$ using $O\left( n\left( \log n + \log\frac{1}{\eps} + \log\log\frac{1}{\delta} \right) + \log\log m \right)$ bits of space.
\end{restatable} 
\begin{proof}
Let $\ell = 6\eps^{-2} \log(6n\delta^{-1})$ and $p = \nfrac{6\ell}{m}$.  On each insertion of a vote $v$, select $v$ with probability $p$ and store for every $i \in [n]$, the number of candidates that candidate $i$ beats in the vote $v$. Keep these exact counts in a counter of length $n$.

Then it follows from the proof of \Cref{thm:heavy_hitters} that $\ell \le |\SC| \le 11\ell$ with probability at least $(1-{\delta}/{3})$. Moreover, from a straightforward application of the Chernoff bound (see \cite{deysampling}),  it follows that if $\hat{s}(i)$ denotes the Borda score of candidate $i$ restricted to the sampled votes, then:
$$\Pr\left[\forall i \in [n], \left|\frac{m}{|\SC|} \hat{s}(i) - s(i)\right| < \eps m n\right] > 1-\delta$$

The space complexity for exactly storing the counts is $O(n \log (n \ell)) = O(n (\log n + \log \eps^{-1} + \log \log \delta^{-1}))$ and the space for sampling the votes is $O(\log \log m)$ by \cref{lem:sampling_ub}.

\end{proof}

\begin{restatable}{theorem}{ThmMaximinUB}\label{thm:maximin}
 Assume the length of the stream is known beforehand. Then there is a randomized one-pass algorithm $\mathcal{A}$ for \textsc{$(\eps, \varphi)$-List maximin} problem which succeeds with probability at least $1-\delta$ using $O\left(n \eps^{-2} \log^2 n + n\eps^{-2} \log n \log \delta^{-1} + \log\log m \right)$ bits of space.
\end{restatable}
\begin{proof}
  Let $\ell=(\nfrac{8}{\eps^2})\ln(\nfrac{6n}{\delta})$ and $p = \nfrac{6\ell}{m}$. We put the current vote in a set $\SC$ with probability $p$. Then it follows from the proof of \Cref{thm:heavy_hitters} that $\ell \le |\SC| \le 11\ell$ with probability at least $(1-\nfrac{\delta}{3})$. Suppose $|\SC| = \ell_1$; let $\mathcal{S} = \{ v_i : i\in[\ell_1] \}$ be the set of votes sampled. Let $D_\EC(x,y)$ be the total number of votes in which $x$ beats $y$ and $D_\SC(x, y)$) be the number of such votes in $\SC$. Then by the choice of $\ell$ and the Chernoff bound (see \cite{deysampling}), it follows that $|D_\SC(x,y)\nfrac{m}{\ell_1} - D_\EC(x,y)| \le \nfrac{\eps m}{2}$ for every pair of candidates $x, y \in \UC$. Note that each vote can be stored in $O(n\log n)$ bits of space. Hence simply finding $D_\SC(x,y)$ for every $x, y\in \UC$ by storing $\SC$ and returning all the items with maximin score at least $(\phi - \eps/2) \ell_1$ in $\SC$ requires $O\left( n\eps^{-2} \log n (\log n + \log \delta^{-1})  + \log \log m \right)$ bits of memory, with the additive $O(\log \log m)$ due to \cref{lem:sampling_ub}.
\end{proof}

\subsection{Unknown stream length}\label{sec:unknown}

Now we consider the case when the length of the stream is not known beforehand. We present below an algorithm for {\sc $(\eps, \varphi)$-List heavy hitters} and {\sc $\eps$-Maximum} problems in the setting where the length of the stream is not known beforehand.

\begin{theorem}\label{thm:UbUnknownMax}
 There is a randomized one-pass algorithm for {\sc $(\eps, \varphi)$-List heavy hitters} and {\sc $\eps$-Maximum} problems with space complexity $O\left({\eps^{-1}}\log \eps^{-1} + \varphi^{-1}\log n + \log\log m \right)$ bits and update time $O(1)$ even when the length of the stream is not known beforehand.
\end{theorem}

\begin{proof}
 We describe below a randomized one-pass algorithm for the {\sc $(8\eps, \varphi)$-List heavy hitters} problem. We may assume that the length of the stream is at least $\nfrac{1}{\eps^2}$; otherwise, we use the algorithm in \Cref{thm:heavy_hitters} and get the result. Now we guess the length of the stream to be $\nfrac{1}{\eps^2}$, but run an instance $\IC_1$ of \Cref{alg:heavy_hitters} with $\ell=\nfrac{\log(\nfrac{6}{\delta})}{\eps^3}$ at line \ref{alg:ell}. By the choice of the size of the sample (which is $\Theta(\nfrac{\log(\nfrac{1}{\delta})}{\eps^3})$), $\IC_1$ outputs correctly with probability at least $(1-\delta)$, if the length of the stream is in $[\nfrac{1}{\eps^2},\nfrac{1}{\eps^3}]$. If the length of the stream exceeds $\nfrac{1}{\eps^2}$, we run another instance $\IC_2$ of \Cref{alg:heavy_hitters} with $\ell=\nfrac{\log(\nfrac{6}{\delta})}{\eps^3}$ at line \ref{alg:ell}. Again by the choice of the size of the sample, $\IC_2$ outputs correctly with probability at least $(1-\delta)$, if the length of the stream is in $[\nfrac{1}{\eps^3},\nfrac{1}{\eps^4}]$. If the stream length exceeds $\nfrac{1}{\eps^3}$, we discard $\IC_1$, free the space it uses, and run an instance $\IC_3$ of \Cref{alg:heavy_hitters} with $\ell=\nfrac{\log(\nfrac{6}{\delta})}{\eps^3}$ at line \ref{alg:ell} and so on. At any point of time, we have at most two instances of \Cref{alg:heavy_hitters} running. When the stream ends, we return the output of the older of the instances we are currently running. We use the approximate counting method of Morris \citep{morris1978counting} to approximately count the length of the stream. We know that the Morris counter outputs correctly with probability $(1-2^{-\nfrac{k}{2}})$ using $O(\log\log m + k)$ bits of space at any point in time \citep{flajolet1985approximate}. Also, since the Morris counter increases only when an item is read, it outputs correctly up to a factor of four at every position if it outputs correctly at positions $1, 2, 4, \ldots, 2^{\lfloor \log_2 m \rfloor}$; call this event $E$. Then we have $\Pr(E) \ge 1-\delta$ by choosing $k=2\log_2(\nfrac{\log_2 m}{\delta})$ and applying union bound over the positions $1, 2, 4, \ldots, 2^{\lfloor \log_2 m \rfloor}$. The correctness of the algorithm follows from the correctness of \Cref{alg:heavy_hitters} and the fact that we are discarding at most $\eps m$ many items in the stream (by discarding a run of an instance of \Cref{alg:heavy_hitters}). The space complexity and the $O(1)$ update time of the algorithm follow from \Cref{thm:heavy_hitters},  the choice of $k$ above, and the fact that we have at most two instances of \Cref{alg:heavy_hitters} currently running at any point of time.
 
 The algorithm for the {\sc $\eps$-Maximum} problem is same as the algorithm above except we use the algorithm in \Cref{thm:heavy_hitters_weak} instead of \Cref{alg:heavy_hitters}.
\end{proof}

Note that this proof technique does not seem to apply to our optimal \cref{alg:heavy_hitters2}. Similarly to \Cref{thm:UbUnknownMax}, we get the following result for the \longversion{{\sc $\eps$-Minimum, $(\eps,\phi)$-Borda,} and {\sc $(\eps,\phi)$-Maximin}}\shortversion{other} problems.

\begin{theorem}\label{thm:UbUnknownMin}
 There are randomized one-pass algorithms for {\sc $\eps$-Minimum, $(\eps,\phi)$-Borda,} and {\sc $(\eps,\phi)$-Maximin} problems with space complexity $O\left((\nfrac{1}{\eps})\log\log(\nfrac{1}{\eps\delta}) + \log\log m \right)$, $O\left( n\left( \log n + \log\frac{1}{\eps} + \log\log\frac{1}{\delta} \right) + \log\log m \right)$, and $O\left(n\eps^{-2} \log^2 n + n \eps^{-2} \log n \log(\nfrac{1}{\delta}) + \log\log m \right)$ bits respectively even when the length of the stream is not known beforehand. Moreover, the update time for {\sc $\eps$-Minimum} is $O(1)$.
\end{theorem}

\section{Hardness}\label{subsec:lwb}

In this section, we prove space complexity lower bounds for the {\sc $\eps$-Heavy hitters}, {\sc $\eps$-Minimum}, {\sc $\eps$-Borda}, and {\sc $\eps$-maximin} problems. We present reductions from certain communication problems for proving space complexity lower bounds. Let us first introduce those communication problems with necessary results.

\subsection{Communication Complexity}

\begin{definition}(\textsc{Indexing}$_{m,t}$)\\
 Let $t$ and $m$ be positive integers. Alice is given a string $x = (x_1, \cdots, x_t)\in [m]^t$. Bob is given an index $i\in [t]$. Bob has to output $x_i$.
\end{definition}

The following is a well known result~\citep{Kushilevitz}.

\begin{lemma}\label{lem:index}
$\mathcal{R}_\delta^{\text{1-way}}(\textsc{Indexing}_{m,t}) = \Omega(t \log m)$ for constant $\delta\in(0,1)$.
\end{lemma}
\longversion{
\begin{definition}(\textsc{Augmented-indexing}$_{m,t}$)\\
 Let $t$ and $m$ be positive integers. Alice is given a string $x = (x_1, \cdots, x_t)\in [m]^t$. Bob is given an integer $i\in [t]$ and $(x_1, \cdots, x_{i-1})$. Bob has to output $x_i$.
\end{definition}

The following communication complexity lower bound result is due to~\citep{ergun2010periodicity} by a simple extension of the arguments of Bar-Yossef et al \cite{bar2002information}.

\begin{lemma}\label{lem:aug}
$\mathcal{R}_\delta^{\text{1-way}}(\textsc{Augmented-indexing}_{m,t}) = \Omega((1-\delta)t \log m)$ for any $\delta < 1 - \frac{3}{2m}$.
\end{lemma}}

\ignore{
We introduce the following communication problem and show a lower bound.

\begin{definition}(\textsc{Max-sum}$_{m,t}$)\\\label{def:maxsum}
 Alice is given a string $x=(x_1, x_2, \cdots, x_t)\in [m]^t$ of length $t$ over universe $[m]$. Bob is given another string $y=(y_1, y_2, \cdots, y_t)\in [m]^t$ of length $t$ over the same universe $[m]$. The strings $x$ and $y$ is such that the index $i$ that maximizes $x_i+y_i$ is unique. Bob has to output the index $i\in[t]$ which satisfies $x_i+y_i = \max_{j\in[t]}\{x_j+y_j\}$.
\end{definition}

\longversion{We establish the following one way communication complexity lower bound for the \textsc{Max-sum}$_{m,t}$ problem by reducing it from the \textsc{Augmented-indexing}$_{2,t\log m}$ problem.}
\begin{restatable}{lemma}{LemMaxSumLB}\label{mel:maxsum}
 $\mathcal{R}_\delta^{\text{1-way}}(\textsc{Max-sum}_{m,t}) = \Omega(t\log m)$, for every $\delta < \frac{1}{4}$. 
\end{restatable}
}
\citep{SunW15} defines a communication problem called {\sc Perm}, which we generalize to {\sc $\eps$-Perm} as follows.

\begin{definition}({\sc $\eps$-Perm})\\
 Alice is given a permutation $\sigma$ over $[n]$ which is partitioned into $\nfrac{1}{\eps}$ many contiguous blocks. Bob is given an index $i\in[n]$ and has to output the block in $\sigma$ where $i$ belongs.
\end{definition}

Our lower bound for {\sc $\eps$-Perm} matches the lower bound for {\sc Perm} in Lemma $1$ in \citep{SunW15} when $\eps = \nfrac{1}{n}$. For the proof, the reader may find useful some information theory facts described in Appendix A.

\begin{restatable}{lemma}{LemPerm}\label{lem:perm}
 $\RC_\delta^{\text{1-way}} (\eps-\textsc{Perm}) = \Omega(n\log(\nfrac{1}{\eps}))$, for any constant $\delta < \nfrac{1}{10}$.
\end{restatable}
\begin{proof}
 Let us assume $\sigma$, the permutation Alice has, is uniformly distributed over the set of all permutations. Let $\tau_j$ denotes the block the item $j$ is in for $j\in[n]$, $\tau = (\tau_1, \ldots, \tau_n)$, and $\tau_{<j} = (\tau_1, \ldots, \tau_{j-1})$. Let $M(\tau)$ be Alice's message to Bob, which is a random variable depending on the randomness of $\sigma$ and the private coin tosses of Alice. Then we have $\RC^{1-way}(\eps-\textsc{Perm}) \ge H(M(\tau)) \ge I(M(\tau); \tau)$. Hence it is enough to lower bound $I(M(\tau); \tau)$. Then we have the following by chain rule.
 \begin{align*}
 I(M(\tau); \tau) &= \sum_{j=1}^n I(M(\tau); \tau_j | \tau_{<j})\\
 &= \sum_{j=1}^n H(\tau_j | \tau_ < j) - H(\tau_j | M(\tau), \tau_<j)\\
 &\ge \sum_{j=1}^n H(\tau_j | \tau_ < j) - \sum_{j=1}^n H(\tau_j | M(\tau))\\
 &= H(\tau) - \sum_{j=1}^n H(\tau_j | M(\tau))
 \end{align*}
 The number of ways to partition $n$ items into $\nfrac{1}{\eps}$ blocks is  $\nfrac{n!}{((\eps n)!)^{(\nfrac{1}{\eps})}}$ which is $\Omega(\nfrac{(\nfrac{n}{e})^n}{(\nfrac{\eps n}{e})^n})$. Hence we have $H(\tau) = n\log (\nfrac{1}{\eps})$. Now we consider $H(\tau_j | M(\tau))$. By the correctness of the algorithm, Fano's inequality, we have $H(\tau_j | M(\tau)) \le H(\delta) + (\nfrac{1}{10})\log_2((\nfrac{1}{\eps})-1) \le (\nfrac{1}{2}) \log (\nfrac{1}{\eps})$. Hence we have the following.
 \[ I(M(\tau); \tau) \ge (\nfrac{n}{2}) \log (\nfrac{1}{\eps}) \]
\end{proof}

Finally, we consider the \textsc{Greater-than} problem.
\begin{definition}(\textsc{Greater-than}$_{n}$)\\
 Alice is given an integer $x\in [n]$ and Bob is given an integer $y\in [n], y\ne x$. Bob has to output $1$ if $x>y$ and $0$ otherwise.
\end{definition}

The following result is due to \citep{smirnov88, MiltersenNSW98}. We provide a simple proof of it that seems to be missing\footnote{A similar proof appears in \cite{kremer1999randomized} but theirs gives a weaker lower bound.} in the literature.

\begin{restatable}{lemma}{LemGT}\label{lem:gt}
$\mathcal{R}_\delta^{\text{1-way}}(\textsc{Greater-than}_{n}) = \Omega(\log n)$, for every $\delta < \nfrac{1}{4}$. 
\end{restatable}
\begin{proof}
 We reduce the \textsc{Augmented-indexing}$_{2,\lceil\log n\rceil + 1}$ problem to the \textsc{Greater-than}$_{n}$ problem thereby proving the result. Alice runs the \textsc{Greater-than}$_{n}$ protocol with its input number whose representation in binary is $a=(x_1x_2\cdots x_{\lceil\log n\rceil}1)_2$. Bob participates in the \textsc{Greater-than}$_{n}$ protocol with its input number whose representation in binary is $b=(x_1x_2\cdots x_{i-1}1\underbrace{0 \cdots 0}_{(\lceil\log n\rceil-i+1)~ 0's})_2$. Now $x_i=1$ if and only if $a>b.$
\end{proof}

\subsection{Reductions}

We observe that a trivial $\Omega((\nfrac{1}{\varphi})\log n)$ bits lower bound for {\sc $(\eps, \varphi)$-List heavy hitters, $(\eps, \varphi)$-List borda, $(\eps, \varphi)$-List maximin} follows from the fact that any algorithm may need to output $\nfrac{1}{\phi}$ many items from the universe. Also, there is a trivial $\Omega(n \log n)$ lower bound for \textsc{$(\eps, \varphi)$-List borda} and \textsc{$(\eps, \varphi)$-List maximin} because each stream item is a permutation on $[n]$, hence requiring $\Omega(n \log n)$ bits to read.

We show now a space complexity lower bound of $\Omega(\frac{1}{\eps}\log \frac{1}{\phi})$ bits for the \textsc{$\eps$-Heavy hitters} problem.

\longversion{
\begin{restatable}{theorem}{ThmHeavyLbEps}\label{thm:eps_eps}
 Suppose the size of universe $n$ is at least $\nfrac{1}{(\eps\phi^{\mu})}$ for any constant $\mu>0$ and that $\phi > 2 \eps$. Any randomized one pass {\sc $(\eps,\phi)$-Heavy hitters} algorithm with success probability at least $(1-\delta)$ must use $\Omega((\nfrac{1}{\eps})\log \nfrac{1}{\phi})$ bits of space, for constant $\delta\in(0,1)$.
\end{restatable}
\begin{proof}
 Let $\mu>0$ be any constant. Without loss of generality, we can assume $\mu\le 1$. We will show that, when $n\ge\nfrac{1}{(\eps\phi^{\mu})}$, any \textsc{$\eps$-Heavy hitters} algorithm must use $\Omega((\nfrac{1}{\eps})\log\nfrac{1}{\phi})$ bits of memory, thereby proving the result. Consider the \textsc{Indexing}$_{\nfrac{1}{\phi^\mu}, \nfrac{1}{\eps}}$ problem where Alice is given a string $x=(x_1, x_2, \cdots, x_{\nfrac{1}{\eps}})\in [\nfrac{1}{\phi^\mu}]^{\nfrac{1}{\eps}}$ and Bob is given an index $i\in [\nfrac{1}{\eps}]$. The stream we generate is over $[\nfrac{1}{\phi^\mu}]\times[\nfrac{1}{\eps}]\subseteq \UC$ (this is possible since $|\UC|\ge \nfrac{1}{(\eps\phi^{\mu})}$). Alice generates a stream of length $\nfrac{m}{2}$ in such a way that the frequency of every item in $\{(x_j,j) : j\in [\nfrac{1}{\eps}]\}$ is at least $\lfloor{\eps m}/{2}\rfloor$ and the frequency of any other item is $0$. Alice now sends the memory content of the algorithm to Bob. Bob resumes the run of the algorithm by generating another stream of length ${m}/{2}$ in such a way that the frequency of every item in $\{(j,i) : j\in [\nfrac{1}{\phi^\mu}]\}$ is at least $\lfloor{\phi^\mu m}/{2}\rfloor$ and the frequency of any other item is $0$. The frequency of the item $(x_i, i)$ is at least $\lfloor \nfrac{\eps m}{2} + \nfrac{\phi^\mu m}{2}\rfloor$ whereas the frequency of every other item is at most $\lfloor{\phi^\mu m}/{2}\rfloor$. Hence from the output of the {\sc $(\nfrac{\eps}{5},\nfrac{\phi}{2})$-Heavy hitters} algorithm Bob knows $i$ with probability at least $(1-\delta)$. Now the result follows from \Cref{lem:index}.
\end{proof}

We now use the same idea as in the proof of \Cref{thm:eps_eps} to prove an $\Omega(\frac{1}{\eps}\log \frac{1}{\eps})$ space complexity lower bound for the $\eps$-Maximum problem.

\begin{theorem}\label{thm:eps_maximum}
 Suppose the size of universe $n$ is at least $\frac{1}{\eps^{1+\mu}}$ for any constant $\mu>0$. Any randomized one pass {\sc $\eps$-Maximum} algorithm with success probability at least $(1-\delta)$ must use $\Omega(\frac{1}{\eps}\log \frac{1}{\eps})$ bits of space, for constant $\delta\in(0,1)$.
\end{theorem}
\begin{proof}
 Let $\mu>0$ be any constant. Without loss of generality, we can assume $\mu\le 1$. We will show that, when $n\ge\frac{1}{\eps^{1+\mu}}$, any \textsc{$\eps$-Maximum} algorithm must use $\Omega(\frac{1}{\eps}\log\frac{1}{\eps})$ bits of memory, thereby proving the result. Consider the \textsc{Indexing}$_{\nfrac{1}{\eps^\mu}, \nfrac{1}{\eps}}$ problem where Alice is given a string $x=(x_1, x_2, \cdots, x_{\nfrac{1}{\eps}})\in [\nfrac{1}{\eps^\mu}]^{\nfrac{1}{\eps}}$ and Bob is given an index $i\in [\nfrac{1}{\eps}]$. The stream we generate is over $[\nfrac{1}{\eps^\mu}]\times[\nfrac{1}{\eps}]\subseteq \UC$ (this is possible since $|\UC|\ge \frac{1}{\eps^{1+\mu}}$). Alice generates a stream of length $\nfrac{m}{2}$ in such a way that the frequency of every item in $\{(x_j,j) : j\in [\nfrac{1}{\eps}]\}$ is at least $\lfloor{\eps m}/{2}\rfloor$ and the frequency of any other item is $0$. Alice now sends the memory content of the algorithm to Bob. Bob resumes the run of the algorithm by generating another stream of length ${m}/{2}$ in such a way that the frequency of every item in $\{(j,i) : j\in [\nfrac{1}{\eps^\mu}]\}$ is at least $\lfloor{\eps^\mu m}/{2}\rfloor$ and the frequency of any other item is $0$. The frequency of the item $(x_i, i)$ is at least $\lfloor \nfrac{\eps m}{2} + \nfrac{\eps^\mu m}{2}\rfloor$ where as the frequency of every other item is at most $\lfloor{\eps^\mu m}/{2}\rfloor$. Hence the {\sc $\nfrac{\eps}{5}$-Maximum} algorithm must output $(x_i, i)$ with probability at least $(1-\delta)$. Now the result follows from \Cref{lem:index}.
\end{proof}
}

\shortversion{

\begin{restatable}{theorem}{ThmHeavyLbEps}\label{thm:eps_eps}
 Suppose the size of universe $n$ is at least $\nfrac{1}{4\eps(\phi-\eps)}$ and that $\phi > 2\eps$. Any randomized one pass {\sc $(\eps,\phi)$-Heavy hitters} algorithm with success probability at least $(1-\delta)$ must use $\Omega((\nfrac{1}{\eps})\log \nfrac{1}{\phi})$ bits of space, for constant $\delta\in(0,1)$.
\end{restatable}
\begin{proof}
Consider the \textsc{Indexing}$_{\nfrac{1}{2(\phi-\eps)}, \nfrac{1}{2\eps}}$ problem where Alice is given a string $x=(x_1, x_2, \cdots, x_{\nfrac{1}{\eps}})\in [\nfrac{1}{2(\phi-\eps)}]^{\nfrac{1}{2\eps}}$ and Bob is given an index $i\in [\nfrac{1}{2\eps}]$. We assume $\phi > 2\eps$. The stream we generate is over $[\nfrac{1}{2(\phi-\eps)}]\times[\nfrac{1}{2\eps}]\subseteq \UC$ (this is possible since $|\UC|\ge \nfrac{1}{(4\eps(\phi-\eps))}$). 

Let $m$ be a large positive integer. Alice generates a stream of length $m/2$ by inserting $\eps m$ copies of $(x_j, j)$ for each $j \in [\nfrac{1}{2\eps}]$. Alice now sends the memory content of the algorithm to Bob. Bob resumes the run of the algorithm by generating another stream of length $m/2$ by inserting $(\phi - \eps)m$ copies of $(j,i)$ for each $j \in [\nfrac{1}{2(\phi-\eps)}]$. The length of the stream is $m$, the frequency of the item $(x_i, i)$ is $\phi m$, while the frequency of every other item is $(\phi - \eps)m$ or $\eps m$. Hence from the output of the {\sc $(\eps,\phi)$-Heavy hitters} algorithm, Bob knows $i$ with probability at least $(1-\delta)$. Now the result follows from \Cref{lem:index}, since $\frac{1}{\eps} \log\frac{1}{\phi-\eps} > \frac{1}{\eps} \log\frac1{\phi}$. 
\end{proof}

We now use the same idea as in the proof of \Cref{thm:eps_eps} to prove an $\Omega(\frac{1}{\eps}\log \frac{1}{\eps})$ space complexity lower bound for the $\eps$-Maximum problem.

\begin{theorem}\label{thm:eps_maximum}
 Suppose the size of universe $n$ is at least $\frac{1}{\eps^2}$. Any randomized one pass {\sc $\eps$-Maximum} algorithm with success probability at least $(1-\delta)$ must use $\Omega(\frac{1}{\eps}\log \frac{1}{\eps})$ bits of space, for constant $\delta\in(0,1)$.
\end{theorem}
\begin{proof}
 Consider the \textsc{Indexing}$_{\nfrac{1}{\eps}, \nfrac{1}{\eps}}$ problem where Alice is given a string $x=(x_1, x_2, \cdots, x_{\nfrac{1}{\eps}})\in [\nfrac{1}{\eps}]^{\nfrac{1}{\eps}}$ and Bob is given an index $i\in [\nfrac{1}{\eps}]$. The stream we generate is over $[\nfrac{1}{\eps}]\times[\nfrac{1}{\eps}]\subseteq \UC$ (this is possible since $|\UC|\ge \frac{1}{\eps^2}$). Let $m$ be a large positive integer. Alice generates a stream of length $\nfrac{m}{2}$ in such a way that the frequency of every item in $\{(x_j,j) : j\in [\nfrac{1}{\eps}]\}$ is at least $\lfloor{\eps m}/{2}\rfloor$ and the frequency of any other item is $0$. Alice now sends the memory content of the algorithm to Bob. Bob resumes the run of the algorithm by generating another stream of length ${m}/{2}$ in such a way that the frequency of every item in $\{(j,i) : j\in [\nfrac{1}{\eps}]\}$ is at least $\lfloor{\eps m}/{2}\rfloor$ and the frequency of any other item is $0$. The frequency of the item $(x_i, i)$ is at least $\lfloor \eps m\rfloor$ where as the frequency of every other item is at most $\lfloor{\eps m}/{2}\rfloor$. Hence the {\sc $\nfrac{\eps}{5}$-Maximum} algorithm must output $(x_i, i)$ with probability at least $(1-\delta)$. Now the result follows from \Cref{lem:index}.
\end{proof}
}

For {\sc $\eps$-Minimum}, we prove a space complexity lower bound of $\Omega(\nfrac{1}{\eps})$  bits.

\begin{restatable}{theorem}{ThmVetoLB}\label{thm:veto_lb}
 Suppose the universe size $n$ is at least $\nfrac{1}{\eps}$. Then any randomized one pass \textsc{$\eps$-Minimum} algorithm must use $\Omega(\nfrac{1}{\eps})$  bits of space.
\end{restatable}
\begin{proof}
 We reduce from \textsc{Indexing}$_{2,\nfrac{5}{\eps}}$ to \textsc{$\eps$-Minimum} thereby proving the result. Let the inputs to Alice and Bob in \textsc{Indexing}$_{2,\nfrac{5}{\eps}}$ be $(x_1, \ldots, x_{\nfrac{5}{\eps}}) \in \{0,1\}^{\nfrac{5}{\eps}}$ and an index $i\in[\nfrac{5}{\eps}]$ respectively. Alice and Bob generate a stream $\SC$ over the universe $[(\nfrac{5}{\eps})+1]$. Alice puts two copies of item $j$ in $\SC$ for every $j\in\UC$ with $x_j=1$ and runs the \textsc{$\eps$-Minimum} algorithm. Alice now sends the memory content of the algorithm to Bob. Bob resumes the run of the algorithm by putting two copies of every item in $\UC\setminus\{i,(\nfrac{5}{\eps})+1\}$ in the stream $\SC$. Bob also puts one copy of $(\nfrac{5}{\eps})+1$ in $\SC$. Suppose the size of the support of $(x_1, \ldots, x_{\nfrac{5}{\eps}})$ be $\ell$. Since $\nfrac{1}{(2\ell+(\nfrac{2}{\eps})-1)} > \nfrac{\eps}{5}$, we have the following. If $x_i=0$, then the \textsc{$\eps$-Minimum} algorithm must output $i$ with probability at least $(1-\delta)$. If $x_i=1$, then the \textsc{$\eps$-Minimum} algorithm must output $(\nfrac{5}{\eps})+1$ with probability at least $(1-\delta)$. Now the result follows from \Cref{lem:index}.
\end{proof}

We show next a $\Omega(n\log (\nfrac{1}{\eps}))$ bits space complexity lower bound for {\sc $\eps$-Borda}.

\begin{theorem}\label{thm:lwb_borda}
 Any one pass algorithm for {\sc $\eps$-Borda} must use $\Omega(n\log (\nfrac{1}{\eps}))$ bits of space.
\end{theorem}

\begin{proof}
 We reduce {\sc $\eps$-Perm} to {\sc $\eps$-Borda}. Suppose Alice has a permutation $\sigma$ over $[n]$ and Bob has an index $i\in[n]$. The item set of our reduced election is $\UC = [n]\sqcup\DC$, where $\DC = \{d_1, d_2, \ldots, d_{2n}\}$. Alice generates a vote $\vv$ over the item set $\UC$ from $\sigma$ as follows. The vote $\vv$ is $\BC_1 \succ \BC_2 \succ \cdots \succ \BC_{\nfrac{1}{\eps}}$ where $\BC_j$ for $j=1, \ldots, \nfrac{1}{\eps}$ is defined as follows.
 \begin{align*}
 \BC_j &= d_{(j-1)2\eps n + 1} \succ d_{(j-1)2\eps n + 2} \succ \cdots \succ d_{(2j-1)\eps n} \\
 &\succ \sigma_{j\eps n + 1} \succ \cdots \succ \sigma_{(j+1)\eps n} \succ d_{(2j-1)\eps + 1} \succ \cdots \succ d_{2j\eps n}
 \end{align*}

 Alice runs the {\sc $\eps$-Borda} algorithm with the vote $\vv$ and sends the memory content to Bob. Let $\DC_{-i} = \DC \setminus \{i\}$, $\overrightarrow{\DC_{-i}}$ be an arbitrary but fixed ordering of the items in $\DC_{-i}$, and $\overleftarrow{\DC_{-i}}$ be the reverse ordering of $\overrightarrow{\DC_{-i}}$. Bob resumes the algorithm by generating two votes each of the form $i\succ \overrightarrow{\DC_{-i}}$ and $i\succ \overleftarrow{\DC_{-i}}$. Let us call the resulting election $\EC$. The number of votes $m$ in $\EC$ is $5$. The Borda score of the item $i$ is at least $12n$. The Borda score of every item $x\in\UC$ is at most $9n$. Hence for $\eps < \nfrac{1}{15}$, the {\sc $\eps$-Borda} algorithm must output the item $i$. Moreover, it follows from the construction of $\vv$ that an $\eps mn$ additive approximation of the Borda score of the item $i$ reveals the block where $i$ belongs in the {\sc $\eps$-Perm} instance.
\end{proof}

We next give a nearly-tight lower bound for the {\sc $\eps$-maximin} problem.
\begin{theorem}\label{thm:mmlb}
Any one-pass algorithm for {\sc $\eps$-maximin} requires $\Omega(n/\eps^2)$ memory bits of storage.
\end{theorem}

\begin{proof}
We reduce from {\sc Indexing}. Let $\gamma = 1/\eps^2$. Suppose Alice has a string $y$ of length $(n-\gamma)\cdot \gamma$, partitioned into $n-\gamma$ blocks of length $\gamma$ each. Bob has an index $\ell = i + (j-\gamma-1)\cdot \gamma$ where $i \in [\gamma], j \in \{\gamma+1, \dots, n\}$. The {\sc Indexing} problem is to return $y_\ell$ for which there is a $\Omega(|y|) = \Omega(n/\eps^2)$ lower bound (\Cref{lem:index}).

The initial part of the reduction follows the construction in the proof of Theorem 6 in \cite{VWWZ15}, which we encapsulate in the following lemma. 

\begin{lemma}[Theorem 6 in \cite{VWWZ15}]\label{lem:vangucht}
Given $y$, Alice can construct a matrix $P \in \{0,1\}^{n \times \gamma}$ using public randomness, such that if $P^i$ and $P^j$ are the $i$'th and $j$'th rows of $P$ respectively, then with probability at least $2/3$, $\Delta(P^i,P^j) \geq \frac{\gamma}{2} + \sqrt{\gamma}$ if $y_\ell = 1$ and $\Delta(a,b) \leq \frac{\gamma}{2}-\sqrt{\gamma}$ if $y_\ell = 0$.
\end{lemma}

Let Alice construct $P$ according to \Cref{lem:vangucht} and then adjoin the bitwise complement of the matrix $P$ below $P$ to form the matrix $P' \in \{0,1\}^{2n \times \gamma}$; note that each column of $P'$ has exactly $n$ 1's and $n$ 0's. Now, we interpret each row of $P$ as a candidate and each column of $P$ as a vote in the following way: 
for each $v \in [\gamma]$, vote $v$ has the candidates in $\{c : P'_{c,v}=1\}$ in ascending order in the top $n$ positions and the rest of the candidates in ascending order in the bottom $n$ positions. Alice inserts these $\gamma$ votes into the stream and sends the state of the {\sc $\eps$-Maximin} algorithm to Bob as well as the Hamming weight of each row in $P'$. Bob inserts $\gamma$ more votes, in each of which candidate $i$ comes first, candidate $j$ comes second, and the rest of the $2n-2$ candidates are in arbitrary order.

Note that because of Bob's votes, the maximin score of $j$ is the number of votes among the ones casted by Alice in which $j$ defeats  $i$. Since $i < j$, in those columns $v$ where $P_{i,v} = P_{j,v}$, candidate $i$ beats candidate $j$. Thus, the set of votes in which $j$ defeats $i$ is $\{v \mid P_{i,v} = 0, P_{j,v}=1\}$. The size of this set is $\frac{1}{2}\left(\Delta(P^i, P^j) + |P^j| - |P^i|\right)$. Therefore, if Bob can estimate the maximin score of $j$ upto $\sqrt{\gamma}/4$ additive error, he can find $\Delta(P^i,P^j)$ upto $\sqrt{\gamma}/2$ additive error as Bob knows $|P^i|$ and $|P^j|$. This is enough, by \Cref{lem:vangucht}, to solve the {\sc Indexing} problem with probability at least $2/3$.
\end{proof}

Finally, we show a space complexity lower bound that depends on the length of the stream $m$.

\begin{restatable}{theorem}{ThmLogLogn}\label{thm:loglogn}
 Any one pass algorithm for {\sc $\eps$-Heavy hitters}, {\sc $\eps$-Minimum}, {\sc $\eps$-Borda}, and {\sc $\eps$-maximin} must use $\Omega(\log \log m)$ memory bits, even if the stream is over a universe of size $2$, for every $\eps < \frac{1}{4}$.
\end{restatable}

\begin{proof}
 It is enough to prove the result only for {\sc $\eps$-Heavy hitters} since the other three problems reduce to {\sc $\eps$-Heavy hitters} for a universe of size $2$. Suppose we have a randomized one pass {\sc $\eps$-Heavy hitters} algorithm which uses $s(m)$ bits of space. Using this algorithm, we will show a communication protocol for the \textsc{Greater-than}$_{m}$ problem whose communication complexity is $s(2^m)$ thereby proving the statement. The universal set is $\UC=\{0,1\}$. Alice generates a stream of $2^x$ many copies of the item $1$. Alice now sends the memory content of the algorithm. Bob resumes the run of the algorithm by generating a stream of $2^y$ many copies of the item $0$. If $x>y$, then the item $1$ is the only $\eps$-winner; whereas if $x<y$, then the item $0$ is the only $\eps$-winner.
\end{proof}

\section*{Acknowledgments}
We would like to thank Jelani Nelson for a helpful conversation which led us to discover an error in a previous version of the paper.

\newpage

\bibliographystyle{alpha}
\bibliography{sketch}

\newpage

\section*{Appendix A}

\subsection*{Information Theory Facts}\label{app:main}
For a discrete random variable $X$ with possible values $\{x_1,x_2,\ldots,x_n\}$, the
Shannon entropy of $X$ is defined as
$H(X)=-\sum_{i=1}^{n}\Pr(X=x_i)\log_2 \Pr(X=x_i)$. Let $H_b(p)=-p\log_2 p-(1-p)\log_2
(1-p)$ denote the binary entropy function when $p \in (0,1)$.
For two random variables $X$ and $Y$ with possible values $\{x_1,x_2,\ldots,x_n\}$ and
$\{y_1,y_2,\ldots,y_m\}$, respectively, the conditional entropy of $X$ given $Y$ is defined
as $H(X\ |\ Y)=\sum_{i,j}\Pr(X=x_i,Y=y_j)\log_2\frac{\Pr(Y=y_j)}{\Pr(X=x_i,Y=y_j)}$.
Let $I(X; Y) = H(X) - H(X\ |\ Y) = H(Y) - H(Y\ |\ X)$
denote the mutual information between two random variables $X, Y$.
Let $I(X; Y\ |\ Z)$ denote the mutual information between two random variables $X, Y$
conditioned on $Z$, i.e.,
$I(X ; Y\ |\ Z) = H(X\ |\ Z) - H(X\ |\ Y, Z)$.
The following summarizes several basic properties of entropy and mutual information.
\begin{proposition}\label{prop:mut}
Let $X, Y, Z, W$ be random variables.
\begin{enumerate}
\item If $X$ takes value in $\{1,2, \ldots, m\}$, then $H(X) \in [0, \log m]$.
\item $H(X) \ge H(X\ |\ Y)$ and $I(X; Y) = H(X) - H(X\ |\ Y) \ge 0$.
\item If $X$ and $Z$ are independent, then we have $I(X; Y\ |\ Z) \ge I(X; Y)$.
Similarly, if $X, Z$ are independent given $W$, then $I(X; Y\ |\ Z, W) \ge I(X; Y\ |\ W)$.
\item (Chain rule of mutual information)
$I(X, Y; Z) = I(X; Z) + I(Y; Z\ |\ X).$
More generally, for any random variables $X_1, X_2, \ldots, X_n, Y$,
$\textstyle I(X_1, \ldots, X_n; Y) = \sum_{i = 1}^n I(X_i; Y\ |\ X_1, \ldots, X_{i-1})$.
Thus, $I(X, Y; Z\ |\ W) \ge I(X; Z\ |\ W)$.
\item (Fano's inequality) Let $X$ be a random variable chosen from domain $\mathcal{X}$
according to distribution $\mu_X$, and $Y$ be a random variable chosen from domain
$\mathcal{Y}$ according to distribution $\mu_Y$. For any reconstruction function $g :
\mathcal{Y} \to \mathcal{X}$ with error $\delta_g$,
$$H_b(\delta_g) + \delta_g \log(|\mathcal{X}| - 1) \ge H(X\ |\ Y).$$
\end{enumerate}
\end{proposition}

We refer readers to \citep{cover2012elements} for a nice introduction to information theory.

\section*{Appendix B}

We remark that the algorithm in \Cref{lem:sampling_ub} has optimal space complexity as shown in \Cref{lem:sampling_lb} below which may be of independent interest.\longversion{ We also note that every algorithm needs to toss a fair coin at least $\Omega(\log m)$ times to perform any task with probability at least $\nfrac{1}{m}$.}

\begin{restatable}{proposition}{PropSamplingLB}[$\star$]\label{lem:sampling_lb}
 Any algorithm that chooses an item from a set of size $n$ with probability $p$ for $0< p \le \frac{1}{n}$, in unit cost RAM model must use $\Omega(\log\log m)$ bits of memory. 
\end{restatable}

\begin{proof}
 The algorithm generates $t$ bits uniformly at random (the number of bits it generates uniformly at random may also depend on the outcome of the previous random bits) and finally picks an item from the say $x$. Consider a run $\mathcal{R}$ of the algorithm where it chooses the item $x$ with {\em smallest number of random bits getting generated}; say it generates $t$ random bits in this run $\mathcal{R}$. This means that in any other run of the algorithm where the item $x$ is chosen, the algorithm must generate at least $t$ many random bits. Let the random bits generated in $\mathcal{R}$ be $r_1, \cdots, r_t$. Let $s_i$ be the memory content of the algorithm immediately after it generates $i^{th}$ random bit, for $i\in [t]$, in the run $\mathcal{R}$. First notice that if $t < \log_2 n$, then the probability with which the item $x$ is chosen is more than $\frac{1}{n}$, which would be a contradiction. Hence, $t \ge \log_2 n$. Now we claim that all the $s_i$'s must be different. Indeed otherwise, let us assume $s_i = s_j$ for some $i<j$. Then the algorithm chooses the item $x$ after generating $t-(j-i)$ many random bits (which is strictly less than $t$) when the random bits being generated are $r_1, \cdots, r_i, r_{j+1}, \cdots, r_t$. This contradicts the assumption that the run $\mathcal{R}$ we started with chooses the item $x$ with smallest number of random bits generated.
\end{proof}

\end{document}